\documentclass[11pt]{article}
\usepackage[colorlinks=true,linkcolor=darkblue,citecolor=darkblue]{hyperref}
\usepackage{fullpage,times,fourier,charter,graphicx,color}
\usepackage{amsmath,amssymb,amsthm,enumerate,tocloft}
\usepackage{cleveref}
\crefformat{equation}{(#2#1#3)}
\crefmultiformat{equation}{(#2#1#3)}{ and~(#2#1#3)}{, (#2#1#3)}{ and~(#2#1#3)}

\def\maketitle{\par\noindent{\color{darkblue}\Large\bf\sffamily\thetitle\par}\vglue1.6ex
\noindent{\large\theauthor}\\[1ex]
\textit{\theaddress}\\[0.7ex]
{\small\today}\par\vglue1.4\bigskipamount}
\def\title#1{\def\thetitle{#1}}
\def\author#1{\def\theauthor{#1}}
\def\address#1{\def\theaddress{#1}}
\allowdisplaybreaks
\definecolor{darkred}{rgb}{0.9,0,0}
\definecolor{darkgray}{rgb}{0.4,0.4,0.4}
\definecolor{darkblue}{rgb}{0,0,0.7}
\definecolor{brown}{rgb}{0.6,0.2,0.2}

\let\eps=\varepsilon
\def\ii{\mathrm{i}}

\makeatletter

\renewcommand\section{\@startsection {section}{1}{\z@}%
  {-2.2ex \@plus -1ex \@minus -.2ex}%
  {1ex \@plus.1ex}%
  {\normalfont\bf\sffamily\color{darkblue}}}
\renewcommand\subsection{\@startsection{subsection}{2}{\z@}%
  {-1.6ex\@plus -0.4ex \@minus -.2ex}%
  {0.6ex \@plus .1ex}%
  {\normalfont\small\bf\sffamily\color{black}}}
\renewcommand\subsubsection{\@startsection{subsubsection}{3}{\z@}%
  {-0.6ex\@plus -0.2ex \@minus -.2ex}%
  {0.4ex \@plus .1ex}%
  {\normalfont\normalsize\it}}
\renewcommand\paragraph{\@startsection{paragraph}{4}{\z@}%
  {0.2ex \@plus0.2ex \@minus0.1ex}{-0.5em}%
  {\normalfont\normalsize\bfseries}}
\makeatother

\def\fbf#1{\setbox0=\hbox{$#1$}\kern-0.10\wd0
  \lower0.02em\copy0\kern-\wd0 \lower0.02em\hbox{\kern+0.05em\copy0}\kern-\wd0
  \raise0.02em\copy0\kern-\wd0 \raise0.02em\hbox{\kern-0.05em\box0}}
\let\truenabla=\nabla
\def\nabla{{\fbf\truenabla}}

\makeatletter
\let\tru@int=\int
\def\int{\mathop{\textstyle\tru@int}\limits}
\def\overl@ss#1#2{\vcenter{\offinterlineskip
        \ialign{$\m@th#1\hfil##\hfil$\crcr#2\crcr<\crcr } }}
\def\overgr@at#1#2{\vcenter{\offinterlineskip
        \ialign{$\m@th#1\hfil##\hfil$\crcr#2\crcr>\crcr } }}
\def\gl{\mathrel{\mathpalette\overl@ss>}}
\def\lg{\mathrel{\mathpalette\overgr@at<}}

\def\Real{\mathbb{R}}
\def\Complex{\mathbb{C}}

\def\Im{\mathop{\rm Im}\nolimits}

\def\partialvint{\int\kern-0.94em-\kern0.2em}

\let\^=\hat
\let\==\bar
\let\~=\tilde
\let\@=\mathbf
\let\_=\mathsf

\let\le=\leqslant
\let\ge=\geqslant

\def\e{\mathrm{e}}
\def\D{\mathcal{D}}
\def\P{\mathcal{P}}
\def\Q{\mathcal{Q}}
\def\R{\mathcal{R}}

\def\t{{\mathrm{t}}}

\def\diag{\mathop{\rm diag}\nolimits}

\def\half{{\textstyle\frac12}}

\makeatother

\advance\parskip 4pt
\advance\textwidth -2em
\advance\hoffset 1em
\advance\textheight 7mm
\advance\voffset -2mm

\def\be{\begin{equation}}
\def\ee{\end{equation}}
\def\bse{\begin{subequations}}
\def\ese{\end{subequations}}

\numberwithin{equation}{section}
\theoremstyle{plain}
\newtheorem{theorem}{Theorem}[section]
\newtheorem{lemma}[theorem]{Lemma}
\newtheorem{proposition}[theorem]{Proposition}
\newtheorem{corollary}[theorem]{Corollary}
\newtheorem{remark}[theorem]{Remark}
\newtheorem{definition}[theorem]{Definition}

2

\begin{document}
\title{On the dispersionless limit of the Manakov system, its Riemann invariants,
and the modulational stability of its counterpropagating plane waves}
\author{Jimmie Adriazola$^1$ and Gino Biondini$^2$}

\address{
\normalsize\it $^1$: School of Mathematical and Statistical Sciences, Arizona State University, Tempe, AZ\\
\normalsize\it $^2$: Department of Mathematics, State University of New York at Buffalo, Buffalo, NY}

\maketitle

\begin{quotation}
\noindent\small\textbf{Abstract.}
We study the dispersionless limit of the Manakov system, the integrable two-component generalization of the nonlinear Schrödinger equation. We derive the resulting four-component genus-zero Manakov-Whitham system, characterize its hydrodynamic structure, and show that it passes the Haantjes tensor test for integrability. We show that the branch points of the spectral curve associated with plane wave solutions of the Manakov system are the local Riemann invariants of the dispersionless system.
We also use the characteristic speeds to classify the baseband modulational stability/instability of the plane waves, and we study a direct linearization of the Manakov system to characterize their finite-wavenumber stability and verify agreement with the Whitham prediction in the long-wave limit.
Finally, we validate the predictions by comparing them with the results of direct numerical simulations.
\end{quotation}

\bigskip

\section{Introduction}

The Manakov system --- namely, the integrable two-component generalization of the celebrated nonlinear Schr\"odinger (NLS) equation, first written down in \cite{Manakov1974} --- is a prototypical system that arises in many different physical contexts such as water waves, nonlinear optics and Bose-Einstein condensates.
The Manakov system is also a completely integrable system, and as such it is amenable to exact analytical treatment.
As a result, it has been studied extensively in the literature
\cite{NMPZ1984,AS1981,APT2004,HasegawaKodama}.
In particular, the initial value problem for the focusing and defocusing Manakov system with zero boundary conditions at infinity was solved via the inverse scattering transform (IST) in \cite{Manakov1974},
which was later generalized to an arbitrary number of components as well as to matrix NLS systems in \cite{APT2004}.
The IST for the defocusing Manakov system with non-zero boundary conditions (NZBC) was formulated in 
\cite{PAB2006} and later revisited in \cite{BK2015}, while the IST for the focusing case was formulated in
\cite{KBK2015}, see also \cite{ABP_JPA2022,ABP_EAJAM2022,CaudrelierZhang2012,CaudrelierZhang2014}.

On the other hand, several fundamental questions concerning the Manakov system remain open.
Prominent examples of this state of affairs are the lack of a satisfactory characterization of its finite-genus solutions, the lack of a satisfactory solution of the initial value problem with periodic boundary conditions, and the lack of a comprehensive Whitham modulation theory,
(i.e., a systematic theory of the slow modulation of its periodic and quasi-periodic wave solutions).
For the scalar NLS equation, the availability of Riemann invariants for the modulation equations
is by now classical~\cite{ElHoefer2016,Kamchatnov,Whitham1974}.
For the Manakov system, however, the situation is considerably more delicate, in part because the
associated spectral curve is a trigonal (rather than hyperelliptic) Riemann surface~\cite{JNLS2000v10p291}.
Nonetheless, considerable progress has been made in a number of important works.
Wright~\cite{Wright2013} obtained a general class of one-phase (genus-one) solutions of the Manakov system by algebro-geometric methods and proposed that the branch points of the associated trigonal spectral curve should play the role of Riemann invariants for the corresponding Whitham system. He confirmed this correspondence in a nontrivial dispersionless (genus-zero) example. Theorem~\ref{t:branchpoints} below establishes a general local result at genus zero: every simple branch point $k_j$ is a local Riemann invariant, with characteristic speed $c_j=\lambda_j-k_j$. The proof follows directly from the exact algebraic identity~\eqref{e:keyidentity}, which shows that the gradient of each simple branch point is a left eigenvector of the dispersionless hydrodynamic matrix.
Independently, Kamchatnov~\cite{Kamchatnov2013EPL} obtained periodic solutions of the two-component system
via the finite-gap integration method within the Ablowitz-Kaup-Newell-Segur (AKNS) scheme, a program subsequently developed in a series of works.
In particular, Kamchatnov~\cite{Kamchatnov2014} established the connection between these two parametrizations
and derived the Whitham equations for the physically important classes of density (in-phase) and polarization
(counter-phase) waves, showing in each case that the modulation equations reduce to Riemann diagonal form.
Related studies of nonlinear polarization waves, dispersive hydrodynamics, and Riemann problems for
two-component Bose--Einstein condensates were carried out by Kamchatnov and collaborators
\cite{KKLP2014,LPK2013,CongyKamchatnovPavloff2016,IvanovKamchatnovCongyPavloff2017},
and the connection between the Manakov system and the Kowalevski equations was elucidated
in~\cite{KamchatnovSokolov2015}.
In~\cite{JNLS2000v10p291}, Forest, McLaughlin, Muraki and Wright derived the linearized dispersion relation for counterpropagating plane waves, identified long-wave and intermediate-wavelength instability regimes, and related the unstable modes to branch and double points of the spectral curve. Simple-wave reductions for two-component Bose--Einstein condensates were subsequently studied by Ivanov and Kamchatnov~\cite{IvanovKamchatnov2018}.

Despite this substantial body of work, a general formulation of Whitham modulation theory for the Manakov system remains by and large an open problem.
And, perhaps surprisingly, several basic questions are still open even for the dispersionless Manakov system, i.e., the system of genus-zero Whitham equations that govern the modulation of the plane wave solutions of the Manakov system, which is the object of the present work.

Specifically, in this work we present several novel results on the dispersionless Manakov system and the modulational stability of the plane-wave solutions of the Manakov system itself. 
First, in section~\ref{s:dispersionless} we derive the four-component dispersionless system and characterize its hydrodynamic structure. In section~\ref{s:Laxpair} we compute the spectral curve associated with the two-component plane waves and obtain explicitly the quartic whose roots are its finite branch points. The main structural result, proved in section~\ref{s:branchpoints}, is Theorem~\ref{t:branchpoints}: every simple branch point $k_j$ is a local Riemann invariant of the dispersionless system, with characteristic speed $c_j=\lambda_j-k_j$. The proof rests on the exact algebraic identity~\eqref{e:keyidentity}, which shows directly that the gradient of a simple branch point is a left eigenvector of the hydrodynamic coefficient matrix. In section~\ref{s:inversion}, after fixing the Galilean frame, we solve the generic inversion problem and prove in Proposition~\ref{p:inversionuniqueness} that the admissible inverse state is unique up to the residual component-exchange symmetry~\eqref{e:residualsymmetry}. In section~\ref{s:modulationalstability} we relate collisions of branch points to collisions of characteristic speeds through the exact discriminant identity~\eqref{e:branch-char-discriminant-identity}, obtain the complete baseband instability region, and prove that the dimensionless growth coefficient is strictly decreasing with $|s_3|$ at fixed counterflow and attains its global maximum $1/(2\sqrt2)$ at $s_3=0$ and $|W|=\sqrt{3/2}$. Finally, in section~\ref{s:linearization} we derive a discriminant criterion for the full finite-wavenumber linearized spectrum, recover the Whitham criterion exactly in the limit $\Xi\to0$, and resolve the degenerate Whitham boundary by showing that it is unstable for arbitrarily small nonzero perturbation wavenumbers. Section~\ref{s:numerics} compares these analytical predictions with direct numerical simulations,
while section~\ref{s:conclusions} ends this work with some final remarks.

\section{The dispersionless (genus-zero) Manakov-Whitham system}
\label{s:dispersionless}

The defocusing Manakov system  \cite{Manakov1974} in semiclassical scaling and dimensionless form is the coupled system of nonlinear partial differential equations
\be
\label{e:Manakov}
\ii \eps\,\@q_t + \tfrac12{\eps^2}\@q_{xx} - (\@q^\dagger\@q)\,\@q = 0,
\qquad
\@q = \@q(x,t)= (q_1,q_2)^\t:\Real\times\Real\mapsto\Complex^2\,,
\ee
where throughout this work the superscripts $\t$ and $\dag$ denote respectively matrix transpose and conjugate transpose and where  
the subscripts $x$ and $t$ denote partial differentiation.
The parameter $\eps>0$ quantifies the relative strength of dispersive effects compared to nonlinear ones, and is the dimensionless counterpart of Planck's constant in the quantum mechanical context.

\paragraph{Two-fluid Madelung ansatz and dispersionless limit of the Manakov system.}
We begin by expressing the solution of the defocusing Manakov system in the Madelung form 
via the two-component polar decomposition
\vspace*{-0.4ex}
\begin{equation}\label{e:ComponentMadelung}
q_j(x,t)=\sqrt{\rho_j(x,t)}\,\e^{\ii \Phi_j(x,t)},
\qquad
j=1,2,
\end{equation}
with $\rho_j\ge 0$.
The semiclassical limit of~\eqref{e:Manakov} is obtained when $\Phi_1$ and~$\Phi_2$ vary over 
spatial scales of $O(\epsilon)$.
Accordingly, one introduces the two component velocities
\vspace*{-0.6ex}
\begin{equation}
v_j(x,t):= \eps\partial_x \Phi_j,\qquad j = 1,2\,.
\end{equation}
In terms of these variables, the real and imaginary parts of the Manakov system \eqref{e:Manakov} yield the hydrodynamic system
\vspace*{-0.2ex}
\begin{gather}
\label{e:ExactCompContinuity}
\rho_{j,t}+(\rho_j v_j)_x=0,
\qquad
j=1,2,
\\
\label{e:ExactCompPhase}
\eps\Phi_{j,t}+\frac{1}{2}v_j^2+\rho
-\frac{\eps^2}{2}\frac{\left(\sqrt{\rho_j}\right)_{xx}}{\sqrt{\rho_j}}=0,
\qquad
j=1,2.
\end{gather}
Equation~\eqref{e:ExactCompContinuity} expresses the component mass conservation laws,
while \eqref{e:ExactCompPhase} encodes the conservation of momentum.
Indeed, differentiating \eqref{e:ExactCompPhase} in \(x\) gives the component momentum laws
\vspace*{-0.6ex}
\begin{equation}\label{e:ExactCompVelocity}
v_{j,t}+v_j v_{j,x}+\rho_x
=
\frac{\eps^2}{2}\partial_x\!\left(\frac{\left(\sqrt{\rho_j}\right)_{xx}}{\sqrt{\rho_j}}\right),
\qquad
j=1,2.
\end{equation}

Up to the extra spatial differentiation needed to derive \eqref{e:ExactCompVelocity},
the system of equations~\eqref{e:ExactCompContinuity} and~\eqref{e:ExactCompVelocity}
is equivalent to the original system~\eqref{e:Manakov}.
The dispersionless limit of the Manakov system is obtained by dropping the \(O(\eps^2)\) quantum pressure terms in~\eqref{e:ExactCompVelocity},
which yields the four-component first-order system
\bse
\label{e:ManakovDispless}
\begin{gather}
\label{e:DispComp1}
\rho_{1,t}+(\rho_1 v_1)_x=0,
\\
\rho_{2,t}+(\rho_2 v_2)_x=0,
\\
\label{e:DispComp2}
v_{1,t}+v_1v_{1,x}+(\rho_1+\rho_2)_x = 0,
\\
v_{2,t}+v_2v_{2,x}+(\rho_1+\rho_2)_x = 0,
\end{gather}
\ese
Equations \eqref{e:ManakovDispless} can also be written compactly in vector form as 
\begin{gather}
\label{e:DisplessManakovMatrixForm}
\@y_t+M(\@y)\,\@y_x=0,
\\
\intertext{with}
\@y=(\rho_1,\rho_2,v_1,v_2)^\t,
\end{gather}
and with the coefficient matrix
\begin{equation}\label{e:AComp}
M(\@y) = 
\begin{pmatrix}
v_1 & 0   & \rho_1 & 0\\
0   & v_2 & 0 & \rho_2\\
1   & 1   & v_1 & 0\\
1   & 1   & 0 & v_2
\end{pmatrix}.
\end{equation}

\begin{remark}
The system~\eqref{e:ManakovDispless}, or equivalently~\eqref{e:DisplessManakovMatrixForm},
is the dispersionless Manakov system, i.e., the system of genus-zero Whitham modulation equations for the Manakov system, and is the main subject of the present work.
\end{remark}

\begin{remark}
The characteristic speeds of the dispersionless Manakov system~\eqref{e:ManakovDispless} are the eigenvalues $c_1,\dots,c_4$ of $M(\@y)$, which are the roots of the first of the two quartics that will appear in the present work, namely, the characteristic polynomial of \(M(\@y)\):
\bse
\label{e:charspeedquartic}
\begin{equation}
P(c;\@y) = 
\det(M-c I_4)=((v_1-c)^2-\rho_1)((v_2-c)^2-\rho_2)-\rho_1\rho_2,
\end{equation}
where $I_n$ is the $n\times n$ identity matrix.
Or, equivalently,
\begin{multline}
P(c;\@y) =
c^4-2(v_1+v_2)c^3+\big(v_1^2+4v_1v_2+v_2^2-\rho_1-\rho_2\big)c^2 
\\
    + 2 \big(\rho_1v_2+\rho_2v_1-v_1^2v_2-v_1v_2^2\big)c
    + \big(v_1^2v_2^2-\rho_1v_2^2-\rho_2v_1^2\big).
\end{multline}
\ese
\end{remark}
The second quartic that will appear in this work arises from the spectral curve associated with the Manakov system, and will make its appearance in section~\ref{s:Laxpair}.
While both of these quartics will play a key role in the analysis, 
it will be important to distinguish between them.

\paragraph{Hamiltonian structure.}~
\unskip It is worthwhile to point out that the dispersionless system~\eqref{e:ManakovDispless} is Hamiltonian with respect to the constant local Poisson operator of hydrodynamic type
\be
\mathcal J
=
-\begin{pmatrix}
0&0&\partial_x&0\\
0&0&0&\partial_x\\
\partial_x&0&0&0\\
0&\partial_x&0&0
\end{pmatrix},
\ee
namely,
\be
\@y_t=\mathcal J\,\frac{\delta H}{\delta\@y},
\qquad
H[\@y]
=
\int
\left[
\frac12\rho_1v_1^2
+\frac12\rho_2v_2^2
+\frac12(\rho_1+\rho_2)^2
\right]\,dx.
\label{e:dispersionlessHamiltonian}
\ee
Indeed,
\[
\frac{\delta H}{\delta\rho_j}
=
\frac12v_j^2+\rho_1+\rho_2,
\qquad
\frac{\delta H}{\delta v_j}
=
\rho_jv_j,
\qquad j=1,2,
\]
which reproduces~\eqref{e:ManakovDispless}. Thus the system belongs to the class of Hamiltonian systems of hydrodynamic type with a nondegenerate Dubrovin--Novikov bracket~\cite{DubrovinNovikov1983}.

\paragraph{Haantjes tensor test and characteristic speed quartic.}
An efficient criterion to test the diagonalizability of a hydrodynamic system
that does not require the computation of the eigenvalues and eigenvectors of the coefficient matrix
was formulated in \cite{Haantjes} and further developed in~\cite{Ferapontov2006,MathAnn2007}, 
involving the vanishing of the Haantjes tensor associated with the coefficient matrix.
Specifically, for strictly hyperbolic systems, \cite{Ferapontov2006,MathAnn2007} give the following theorem 
as a necessary condition for integrability: 

\begin{theorem}(Haantjes)
A hydrodynamic type system with mutually distinct characteristic speeds is diagonalizable 
if and only if the corresponding Haantjes tensor is identically zero.
\end{theorem}

The calculation of the Haantjes tensor requires calculation of the Nijenhuis tensor first. 
The Nijenhuis tensor of a matrix $M^i _j$ is defined as
\be
N^i _{jk}(M) = 
  M^p _j \partial_{y^p}M^i _k 
  - M^p _k \partial_{y^p}M^i _j 
  - M^i _p ( \partial_{y^j}M^p _k -\partial_{y^k}M^p _j)
\ee
where $\partial_{y^k} f = \partial f/\partial{y^k}$. 
In our case, the matrix $M^i _j$ is the coefficient matrix of the system~\eqref{e:DisplessManakovMatrixForm}, 
for which diagonalizability  is being tested.
Once the Nijenhuis tensor is known, the Haantjes tensor can be obtained as
\be
H^i _{jk}(M) = N ^i _{pr}M^p _j M^r _k - N ^p _{jr}M^i _p M^r _k -N ^p _{rk}M^i _p M^r _j + N ^p _{jk}M^i _r M^r _p\,.
\ee
The above test, applied to the system~\eqref{e:DisplessManakovMatrixForm}, yields:

\begin{proposition}
\label{p:haantjes}
All 64 entries of the Haantjes tensor of the dispersionless Manakov system~\eqref{e:DisplessManakovMatrixForm} vanish identically.
\end{proposition}

\begin{proof}
The explicit calculation of tensor entries for the system~\eqref{e:DisplessManakovMatrixForm} is easily carried out with any computer algebra software, and allows one to show that all entries of the Nijenhuis tensor vanish except those with indices (1,2,1), (2,1,2), (3,1,3), (3,2,3), (4,1,4) and~(4,2,4), all of which are equal to~1,
and those with indices (1,1,2), (2,2,1), (3,3,1), (3,3,2), (4,4,1) and (4,4,2), all of which are equal to~$-1$.
Using these values, it is then straightforward to verify that all entries of $H^i_{jk}$ are identically zero.
\end{proof}

Since~\eqref{e:ManakovDispless} is Hamiltonian with respect to a nondegenerate Poisson bracket of hydrodynamic type, Proposition~\ref{p:haantjes} implies semi-Hamiltonianity on the strictly hyperbolic set, and hence integrability in the sense of Tsarev~\cite{DubrovinNovikov1983,Tsarev}. Theorem~\ref{t:branchpoints} below identifies the corresponding local Riemann invariants explicitly as the simple branch points of the spectral curve.

\section{Lax pair, plane wave ansatz, spectral curve and branch points}
\label{s:Laxpair}

\paragraph{Lax pair and Galilean invariance.}
For simplicity, in the calculations of this section we rescale $x$ and $t$ 
(by letting $\~x=x/\eps$ and $\~t=t/\eps$ and dropping tildes), 
which is equivalent to setting $\eps=1$ in~\eqref{e:Manakov}. 
In this normalization, the Manakov system~\eqref{e:Manakov} is the zero-curvature compatibility condition
$X_t-T_x+[X,T]=0$ of the Lax pair
\vspace*{-1ex}
\bse
\label{e:Laxpair}
\begin{gather}
\phi_x=X\phi,
\label{e:scatteringproblem}
\\
\phi_t=T\phi,
\end{gather}
\ese
with $\phi:=\phi(x,t,k):\Real^2\times\Complex\mapsto\Complex^3$ and
\bse
\begin{gather}
X(x,t,k)=\ii kJ+Q,
\qquad
J=\diag(1,-1,-1),
\qquad
Q(x,t)=\begin{pmatrix}0&\@q^\dag\\ \@q&O_2\end{pmatrix},
\\
T(x,t,k)=\ii k^2J+kQ+\tfrac\ii2JQ^2-\tfrac\ii2JQ_x.
\end{gather}
\ese
where $O_2$ is the $2\times2$ zero matrix.
Equation~\eqref{e:scatteringproblem} is referred to as 
the scattering, or spectral, problem; $\phi$ is the eigenfunction, $k$ is the spectral parameter, and $Q$ is the scattering potential.

Like the scalar NLS equation, the Manakov system~\eqref{e:Manakov} is invariant under Galilean transformations, which plays a key role in what follows.
Explicitly, in the normalization $\eps=1$, if $\@q(x,t)$ solves~\eqref{e:Manakov}, then
\begin{equation}
\@q_a(x,t)=\e^{\ii(ax-a^2t/2)}\@q(x-at,t),
\qquad a\in\Real,
\label{e:Galilean}
\end{equation}
is also a solution. At $t=0$, the corresponding potential satisfies
\[
Q_a(x,0)=\e^{-\ii a xJ/2}Q(x,0)\,\e^{\ii a xJ/2}.
\]
It is then easy to see that, if $\phi(x,0,k)$ solves the original scattering problem~\eqref{e:scatteringproblem}, then
$\phi_a(x,0,k-a/2)=\e^{-\ii a xJ/2}\phi(x,0,k)$ solves the scattering problem with the boosted potential~\eqref{e:Galilean}. 
Thus, a Galilean boost sends
\[
v_j\mapsto v_j+a,
\qquad
k\mapsto k-\tfrac12 a.
\]
The corresponding transformation for the branch-point pair $(k,\lambda)$ that will be introduced in section~\ref{s:branchpoints} is
\[
(k_j,\lambda_j)\mapsto (k_j-\tfrac12 a,\lambda_j+\tfrac12 a),
\]
so that $c_j=\lambda_j-k_j$ transforms as $c_j\mapsto c_j+a$.
We take
advantage of this invariance when studying the inversion problem in section~\ref{s:inversion} as well as when carrying out the modulational stability analysis of the plane wave solutions of the Manakov
system in section~\ref{s:modulationalstability}.

\paragraph{Spectral curve of two-component plane waves and branch points.}
We now seek to connect the genus-zero Whitham modulation system to the spectral problem, and to this end, we analyze the Lax spectrum associated with plane wave solutions of the Manakov system.
If $\rho_1,\rho_2,v_1,v_2$ are independent of $x$ and $t$,
the two-component plane wave potential is given by
\be
\@q_\mathrm{pw}(x,t)=
\begin{pmatrix}
\sqrt{\rho_1}\,\e^{\ii(v_1x-\omega_1t)}\\
\sqrt{\rho_2}\,\e^{\ii(v_2x-\omega_2t)}
\end{pmatrix},
\qquad
\omega_j=\tfrac12v_j^2+\rho,
\qquad
\rho=\rho_1+\rho_2,
\label{e:planewaveansatz}
\ee
where $\rho_1,\rho_2\ge0$ and $v_1,v_2\in\Real$ are constants. 
Note that this is an exact solution of~\eqref{e:Manakov} under the normalization $\eps=1$.
To compute the Lax spectrum associated with this plane-wave ansatz, 
we solve the scattering problem~\eqref{e:scatteringproblem} with the corresponding matrix potential
\be
Q_\mathrm{pw}(x,t)=
\begin{pmatrix}
0 & \sqrt{\rho_1}\,\e^{-\ii(v_1x-\omega_1t)} & \sqrt{\rho_2}\,\e^{-\ii(v_2x-\omega_2t)}\\
\sqrt{\rho_1}\,\e^{\ii(v_1x-\omega_1t)} & 0 & 0\\
\sqrt{\rho_2}\,\e^{\ii(v_2x-\omega_2t)} & 0 & 0
\end{pmatrix}.
\ee

The resulting equations are a first-order system of variable coefficient ODEs,
which is generally difficult to solve in closed form.
However, the problem can be simplified considerably.
Indeed, since the spectrum of the scattering problem is independent of time, we may restrict our analysis without loss of generality to the initial potential. At that initial time, the plane-wave potential admits the factorization
\[
Q_\mathrm{pw}(x,0)=\e^{\ii\Sigma x}Q_o\e^{-\ii\Sigma x},
\qquad
Q_o=
\begin{pmatrix}
0&\sqrt{\rho_1}&\sqrt{\rho_2}\\
\sqrt{\rho_1}&0&0\\
\sqrt{\rho_2}&0&0
\end{pmatrix},\qquad
\Sigma=\diag(0,v_1,v_2)\,,
\]
and the corresponding change of dependent variable removes the explicit spatial dependence from the scattering problem
\[
\widetilde\phi(x,0,k)=\e^{-\ii\Sigma x}\phi(x,0,k).
\]
The transformed scattering problem~\eqref{e:scatteringproblem}, therefore, takes the constant-coefficient form
\[
\widetilde\phi_x=X_o\widetilde\phi,
\qquad
X_o=\ii kJ-\ii\Sigma+Q_o,
\]
which is much easier to solve.

By seeking solutions in the form $\widetilde\phi(x,0,k)=\e^{-\ii\lambda x}\@f_o$, one finds
\be
F(k,\lambda;\@y)=0,
\label{e:spectralcurve}
\ee
where $\@y=(\rho_1,\rho_2,v_1,v_2)^\t$ and
\be
\begin{aligned}
F(k,\lambda;\@y)
&:=-\ii\det(X_o+\ii\lambda I_3)\\
&=(k-\lambda+v_2)\rho_1+(k-\lambda+v_1)\rho_2
 -(k+\lambda)(k-\lambda+v_1)(k-\lambda+v_2).
\end{aligned}
\label{e:spectralcubic2}
\ee
Equation~\eqref{e:spectralcurve} defines the spectral curve of the scattering problem for the two-component plane wave ansatz~\eqref{e:planewaveansatz}. In general, \eqref{e:spectralcurve} defines a three-sheeted covering of the complex $k$~plane. 
The explicit expressions for the three roots $\lambda_1,\dots,\lambda_3$ as functions of $k$ are quite complicated (the output is omitted for brevity).
Here, however, the explicit expressions will not be needed, since we only require the branch points of the spectral curve. 

To this end, it is convenient to introduce the symmetrized variables $\rho$, $s_3$, $v$, and $w$ defined by
\bse
\be
\rho=\rho_1+\rho_2,
\qquad
s_3=\frac{\rho_1-\rho_2}{\rho_1+\rho_2},
\qquad
v=v_1+v_2,
\qquad
w=v_1-v_2,
\label{e:zfromy}
\ee
so that
\be
\rho_1=\half\rho(1+s_3),
\qquad
\rho_2=\half\rho(1-s_3),
\qquad
v_1=\half(v+w),
\qquad
v_2=\half(v-w).
\label{e:yfromz}
\ee
\ese
Note from the above definitions that $-1\le s_3\le 1$.  
The parameter $s_3$, which quantifies the relative energy split between the two components, is the third element of the vector $\^s = (s_1,s_2,s_3)^\t$ defined as $s_j = \@q^\dag\sigma_j\@q/(\@q^\dag\@q)$
(where $\sigma_1,\sigma_2,\sigma_3$ are the Pauli matrices with $\sigma_3=\diag(1,-1)$), 
which is referred to as the Stokes polarization vector in optics \cite{Damask} and as the Bloch vector in 
quantum mechanics \cite{NielsenChuang}.
The mean wavenumber is $v/2$, whereas $w$ is the relative wavenumber, or counterflow
between the two components.

\begin{theorem}
Suppose that $w\ne0$ and that the discriminant polynomial defined below has four simple roots. Then the finite simple branch points of the spectral curve $F(k,\lambda,\@y)=0$ are precisely the roots $k_1,\dots,k_4$ of the polynomial
\begin{equation}
\label{e:Deltadef}
\Delta(k,\@y) = a_4 k^4 + a_3 k^3 + a_2 k^2 + a_1 k + a_0\,,
\end{equation}
with
\bse
\label{e:acoeffs}
\begin{align}
a_4 &= 16 w^2,\\
a_3 &= 16w\,(vw-s_3\rho),\\
a_2 &= 4 \rho ^2 - 4 \rho (3 s_3 v w + 5 w^2) + 6 v^2 w^2 - 2w^4\,,\\
a_1 &= \rho ^2 \left(18 s_3 w+2 v\right) - \rho\,(10 v w^2 + 3 s_3 v^2 w - 9 s_3 w^3 ) + v ^3 w^2 - v w^4\,,
\\
a_0 & = -4 \rho ^3 + \tfrac14 \rho^2\big(v^2+12 w^2 + 9 s_3 w(2 v-3 s_3 w)\big)
\nonumber\\&\kern8em{}
  - \tfrac14 \rho w \left(5 v^2 w + s_3(v^3-9 v w^2) + 3 w^3\right)
  + \tfrac1{16} (v^4 w^2 - 2 v^2 w^4+w^6)\,,
\end{align}
\ese
\end{theorem}

\begin{proof}
For fixed $k$, equation~\eqref{e:spectralcurve} is cubic in $\lambda$. A multiple value of $\lambda$ therefore occurs precisely when the discriminant of this cubic vanishes (e.g., see \cite{Bliss1933,Miranda1995}). A direct calculation shows that this discriminant is the polynomial $\Delta(k,\@y)$ in~\eqref{e:Deltadef}. Since $w\ne0$, its leading coefficient is $a_4=16w^2\ne0$, so $\Delta(k,\@y)$ is genuinely quartic.

Let $k_j$ be a simple root of $\Delta$. The cubic equation $F(k_j,\lambda,\@y)=0$ then has exactly one double root $\lambda_j$ and one simple root. Thus,
\be
F(k_j,\lambda_j,\@y)=0,\qquad
F_\lambda(k_j,\lambda_j,\@y)=0,\qquad
F_{\lambda\lambda}(k_j,\lambda_j,\@y)\ne0.
\ee
Moreover, the simplicity of $k_j$ as a zero of the discriminant implies
\be
F_k(k_j,\lambda_j,\@y)\ne0.
\ee
Hence $(k_j,\lambda_j)$ is a simple ramification point of the spectral curve. Therefore the four simple roots of $\Delta(k,\@y)$ are precisely the four finite simple branch points.
\end{proof}

\begin{remark}
The discriminant $\Delta(k,\@y)$ is the second quartic arising in this work. On the generic set $w\ne0$ where its roots are simple, the spectral curve therefore has four finite simple branch points. At special parameter values, such as $w=0$, $s_3=\pm1$, or on the discriminant locus of $\Delta$ itself, the degree or multiplicity structure can degenerate and the corresponding points must be treated separately. The existence of four generic branch points is the first hint that they could be related to the local Riemann invariants of the dispersionless Manakov system~\eqref{e:ManakovDispless}.
\end{remark}

In order to compare the above results to the spectral calculations of Baronio et al. in~\cite{Baronio2014},  it is useful to consider the symmetric frame $v=0$. 
Their background parameters $a_1,a_2,q$ are related to the variables used in this work via the identities 
\begin{equation}
\rho=a_1^2+a_2^2,\qquad
s_3=\frac{a_1^2-a_2^2}{a_1^2+a_2^2},\qquad
w=2q.
\end{equation}
Under this change of variables, the monic polynomial $\widetilde\Delta=\Delta/(16w^2)$ obtained from~\eqref{e:Deltadef} agrees coefficient by coefficient with the symmetric quartic spectral condition obtained in~\cite{Baronio2014}, with the same spectral parameter $k$. Thus, in the symmetric frame, the two quartics are identical. The additional result established below is that each simple root of this quartic is a local Riemann invariant of the dispersionless Manakov system, with characteristic speed determined by Theorem~\ref{t:branchpoints}.

\section{Riemann invariants of the dispersionless Manakov system}
\label{s:branchpoints}

We now prove that, in fact, the branch points of the spectral curve defined by Equation~\eqref{e:spectralcubic2} provide genuine local Riemann invariants of the dispersionless Manakov system.

\begin{theorem}
\label{t:branchpoints}
Let $(k_j,\lambda_j)$ be a simple ramification point of the spectral curve, in the sense that
\[
F(k_j,\lambda_j;\@y)=0,
\qquad
F_\lambda(k_j,\lambda_j;\@y)=0,
\qquad
F_{\lambda\lambda}(k_j,\lambda_j;\@y)\ne0,
\qquad
F_k(k_j,\lambda_j;\@y)\ne0.
\]
Then $k_j$ is a local Riemann invariant for the genus-zero Manakov--Whitham system~\eqref{e:DisplessManakovMatrixForm}.
That is,
\begin{equation}
\label{e:ManakovDispersionlessDiagonal}
(k_j)_t+c_j(k_j)_x=0\,,
\end{equation}
where the corresponding characteristic speed is
\begin{equation}
c_j=\lambda_j-k_j,
\label{e:cjdef}
\end{equation}
and where \((k_j,\lambda_j)\) is the associated simple branch-point pair.
\end{theorem}

\begin{proof}
Let \((k_j,\lambda_j)\) satisfy the simple-branch-point conditions
\begin{equation}
F(k_j,\lambda_j;\@y)=0,
\qquad
F_{\lambda}(k_j,\lambda_j;\@y)=0,
\qquad
F_{\lambda\lambda}(k_j,\lambda_j;\@y)\neq0,\qquad
F_k(k_j,\lambda_j;\@y)\neq 0, 
\end{equation}
The first three conditions imply that \(\lambda_j\) is a double root of the cubic spectral equation over the branch point \(k_j\).

The partial derivatives of \(F\) with respect to $k$, $\lambda$ and $\@y$ are given explicitly by
\bse
\begin{equation}
\begin{aligned}
&\begin{aligned}
F_{\rho_1} &= k-\lambda+v_2, &
F_{\rho_2} &= k-\lambda+v_1,\\
F_{v_1} &= \rho_2-(k-\lambda+v_2)(k+\lambda),\quad &
F_{v_2} &= \rho_1-(k+\lambda)(k-\lambda+v_1),
\end{aligned}
\\[0.3em]
&\begin{aligned}
F_k &= \rho_1+\rho_2-3k^2+2k\lambda-2k(v_1+v_2)+\lambda^2-v_1v_2,\\
F_{\lambda} &= -\rho_1-\rho_2+k^2+2k\lambda-3\lambda^2+2(v_1+v_2)\lambda-v_1v_2.
\end{aligned}
\end{aligned}
\end{equation}
\ese
We obtain from them, by direct computation, the following key algebraic identity:
\begin{equation}
\label{e:keyidentity}
\big(F_{\rho_1},F_{\rho_2},F_{v_1},F_{v_2}\big)\,M
=
(\lambda-k)\big(F_{\rho_1},F_{\rho_2},F_{v_1},F_{v_2}\big)
+\big(-F_{\lambda},-F_{\lambda},F,F\big).
\end{equation}
The identity is exact, and, in particular, at a branch point where \(F=F_{\lambda}=0\), it reduces to
\begin{equation}
\big(F_{\rho_1},F_{\rho_2},F_{v_1},F_{v_2}\big)M
=
(\lambda_j-k_j)\big(F_{\rho_1},F_{\rho_2},F_{v_1},F_{v_2}\big).
\end{equation}

Since \(F(k_j,\lambda_j;\@y)=0\) and \(F_{\lambda}(k_j,\lambda_j;\@y)=0\), implicit differentiation of \(F=0\) gives
\begin{equation}
\partial_{\rho_1} k_j=-\frac{F_{\rho_1}}{F_k},
\qquad
\partial_{\rho_2} k_j=-\frac{F_{\rho_2}}{F_k},
\qquad
\partial_{v_1} k_j=-\frac{F_{v_1}}{F_k},
\qquad
\partial_{v_2} k_j=-\frac{F_{v_2}}{F_k},
\end{equation}
with all of the above quantities evaluated at \((k,\lambda)=(k_j,\lambda_j)\). 
Thus,
\begin{equation}
\nabla_{\@y} k_j =
    -\frac{1}{F_k(k_j,\lambda_j;\@y)}
        \big(F_{\rho_1},F_{\rho_2},F_{v_1},F_{v_2}\big)\Big|_{(k_j,\lambda_j)}.
\end{equation}
Using the exact identity~\eqref{e:keyidentity} and evaluating at the branch point, we obtain
\begin{equation}
\nabla_{\@y} k_j\,M(\@y) =
    (\lambda_j-k_j)\,\nabla_{\@y} k_j.
\end{equation}
Thus, \(\nabla_{\@y} k_j\) is a left eigenvector of \(M(\@y)\), with eigenvalue $c_j$ given by~\eqref{e:cjdef}.

Now differentiate \(k_j(\@y(x,t))\) along solutions of~\eqref{e:DisplessManakovMatrixForm}.
By the chain rule we have 
$(k_j)_t=\nabla_{\@y} k_j\cdot \@y_t$ and
$(k_j)_x=\nabla_{\@y} k_j\cdot \@y_x$.
Using~\eqref{e:DisplessManakovMatrixForm}, we then obtain
\begin{equation}
(k_j)_t = - \nabla_{\@y} k_j\,M(\@y)\,\@y_x = - c_j\,\nabla_{\@y} k_j\,\@y_x = - c_j (k_j)_x\,,
\end{equation}
with $c_j$ given by~\eqref{e:cjdef}.
Hence one obtains~\eqref{e:ManakovDispersionlessDiagonal}, which proves that 
each simple branch point \(k_j\) is indeed a local Riemann invariant.
\end{proof}

\begin{corollary}
\label{c:riemanncoordinates}
On the generic strictly hyperbolic set, the gradients
$\nabla_{\@y}k_1$,
$\nabla_{\@y}k_2$,
$\nabla_{\@y}k_3$ and 
$\nabla_{\@y}k_4$
are linearly independent. Hence $(k_1,k_2,k_3,k_4)$ form a local system of Riemann coordinates for the dispersionless Manakov system.
\end{corollary}

\begin{proof}
By Theorem~\ref{t:branchpoints}, $\nabla_{\@y}k_1$, $\nabla_{\@y}k_2$, $\nabla_{\@y}k_3$, and $\nabla_{\@y}k_4$ are left eigenvectors of $M(\@y)$ associated, respectively, with the characteristic speeds $c_1$, $c_2$, $c_3$, and $c_4$. On the strictly hyperbolic set these four characteristic speeds are distinct, so the corresponding left eigenvectors are linearly independent.
\end{proof}

\begin{remark}
The nondegeneracy condition $F_k(k_j,\lambda_j;\@y)\ne0$ is needed for the implicit-function calculation of $\nabla_{\mathbf y}k_j$. 
When branch points collide, the branch-point coordinates cease to provide four independent local coordinates. 
The diagonal equations obtained above remain meaningful by continuation from a region of simple ramification, but the repeated branch point must be treated as a limiting coordinate. 
The non-strictly-hyperbolic sector of the system is discussed in Lemma~\ref{l:Delta=0}.
\end{remark}

\begin{remark}
The Riemann invariants produced by Theorem~\ref{t:branchpoints} are \textit{local}.
The four roots $k_1,\dots,k_4$ of $\Delta(k,\@y)=0$ are not globally single-valued functions 
of~$\@y$.
They are the branches of an algebraic function, and any labeling of them is only 
consistent on a simply connected domain of the $\@y$~space that avoids the discriminant locus of
$\Delta$.  
Continuation of a labeled quadruple $(k_1,\dots,k_4)$ around a loop encircling that locus 
generically permutes the branches.
Equivalently, as shown in section~\ref{s:inversion}, the inverse map $\@k\mapsto\@y$ is 
multivalued, and its determination requires a branch-selection criterion.
Accordingly, the diagonal form~\eqref{e:ManakovDispersionlessDiagonal} should be understood as 
holding on any region of the $(x,t)$ plane whose image under $\@y$ stays within such a domain.
This is the same state of affairs as for the scalar defocusing NLS equation, and it is sufficient
for the local analysis of simple waves, rarefactions and the characterization of hyperbolicity.
It does, however, mean that global constructions such as the hodograph transform or the solution of
Riemann problems require the branch structure to be tracked explicitly.
\end{remark}

\section{Recovering the physical variables from the Riemann invariants}
\label{s:inversion}

The Riemann invariants $k_1,\dots,k_4$ in terms of the physical variables $\@y = (\rho_1,\rho_2,v_1,v_2)^\t$ are given by the roots of the quartic equation $\Delta(k,\@y)=0$.
Similarly, the characteristic speeds $c_1,\dots,c_4$ are given in terms of $\@y$ by the roots of the
quartic equation~$P(c;\@y)=0$.
In this section, we turn to the so-called inversion problem, namely, we address the question of how one can invert the map $\@y 
\mapsto\@k = (k_1,\dots,k_4)^\t$ and express the physical variables and the characteristic speeds in terms of the Riemann invariants.


To do so, it is convenient to perform the change of dependent variables from $\@y$ to
$\@z = (\rho,s_3,v,w)^\t$, with $\rho$, $s_3$, $v$ and $w$ as in~\eqref{e:zfromy},
which is inverted by~\eqref{e:yfromz}.
Recall that the branch points \(k_1,\dots,k_4\) 
are the four roots of the branch-point quartic \(\Delta(k,\@y)=0\).
For $w\ne0$, since $a_4=16w^2\ne0$, they are equivalently the roots of the monic quartic
\be
\widetilde\Delta(k;\@y):=\frac{\Delta(k,\@y)}{a_4}.
\ee

Since $\widetilde\Delta$ is monic, Vieta's formulas give
\be
\widetilde\Delta(k;\@y) = \prod_{j=1}^4 (k-k_j)
= k^4-S_1k^3+S_2k^2-S_3k+S_4,
\label{e:Delta_S1S4}
\ee
where $S_1(\@k),\dots,S_4(\@k)$ are the symmetric invariants:
\bse
\begin{equation}
\begin{aligned}
S_1(\@k) := k_1+k_2+k_3+k_4
&,\qquad
S_2(\@k) := k_1k_2+k_1k_3+k_1k_4+k_2k_3+k_2k_4+k_3k_4,
\\
S_3(\@k) := k_1k_2k_3+k_1k_2k_4+k_1k_3k_4+k_2k_3k_4
&,\qquad
S_4(\@k) := k_1k_2k_3k_4.
\end{aligned}
\end{equation}
\ese
Comparing the coefficients of the powers of $k$ in~\eqref{e:Delta_S1S4} with those in~\eqref{e:Deltadef}
yields the four equations 
\begin{gather}
\label{e:inversionequations}
S_1(\@k) = - a_3/a_4\,,\qquad 
S_2(\@k) = a_2/a_4\,,\qquad
S_3(\@k) = - a_1/a_4\,,\qquad
S_4(\@k) = a_0/a_4\,,
\end{gather}
with $a_0,\dots,a_4$ as in \eqref{e:acoeffs}.

The task is therefore to solve~\eqref{e:inversionequations} for $\rho$, $s_3$, $v$ and $w$ in terms of $S_1,\dots,S_4$. Here we first fix the Galilean frame and pass to the symmetric frame $v=0$, thereby reducing the inversion problem to the variables $(\rho,s_3,w)$. In this frame, we eliminate $\rho$ and $s_3$ successively and reduce the problem to the determination of $w^2$ from the symmetric invariants $S_1,\dots,S_4$. The full inversion from spectral data in the laboratory frame must also determine the unknown Galilean shift $v$; that problem is treated separately in~\cite{AB2026}.

To start, recall that a Galilean boost, as discussed in Section~\ref{s:Laxpair},
implies $v_j\mapsto v_j+a$, and hence $v=v_1+v_2\mapsto v+2a$. 
It also implies that each branch point transforms as $k_j\mapsto k_j-a/2$. 
Choosing $a=-v/2$ reduces the physical variables to the symmetric frame $v=0$ and translates each branch point by $v/4$. 
We will use this symmetric frame in the remainder of the section, and henceforth $k_j$ and $S_j$ will denote the branch points and their symmetric invariants in this frame.

When $v=0$, the system of equations~\eqref{e:inversionequations} simplifies considerably to the following:%
\begin{equation}
\label{e:inversion-v0}
\begin{aligned}
S_1 = \frac{\rho s_3}{w}
&,\qquad
S_2 = \frac{2\rho^2-10\rho w^2-w^4}{8w^2},
\\
S_3+\frac{3\rho s_3(6\rho+3w^2)}{16w}=0
&,\qquad
S_4 = \frac{(w^2-4\rho)^3-108\rho^2s_3^2w^2}{256w^2}.
\end{aligned}
\end{equation}

The calculations that follow are carried out on the generic branch $w\ne0$ and $s_3\ne0$, for which $S_1=\rho s_3/w\ne0$. 
The zero-counterflow and balanced-density branches require separate limiting formulas and are not covered by the algebra below.
Note that the first of these equations is linear in any of the three remaining dependent variables,
and can therefore be easily solved for each variable.

It appears that the most convenient choice is to solve for~\(\rho\), obtaining
\be
\label{e:rho}
\rho = S_1 w/s_3\,.
\ee
Substituting \eqref{e:rho} into the remaining three equations then yields
\bse
\begin{align}
S_1^2 = s_3^2\left(4S_2+\half w^2\right)+5s_3S_1w
&,
\label{e:2b}
\\
S_1w\left(s_3w+2S_1\right)+\frac{16}{9}s_3S_3 = 0
&,
\label{e:2c}
\\
S_4 = \frac{1}{256}\left[w\left(w-\frac{4S_1}{s_3}\right)^3-108S_1^2w^2\right]
&.
\label{e:2d}
\end{align}
\ese
Again, note that \eqref{e:2c} is linear in~$s_3$.
We therefore continue by solving \eqref{e:2c} for $s_3$, obtaining
\be
s_3 = -\frac{18 S_1^2 w}{16 S_3 + 9 S_1 w^2}\,.
\label{e:s3explicit}
\ee
Finally, substituting the expression \eqref{e:s3explicit} into~\eqref{e:2b} and~\eqref{e:2d} yields respectively
\bse
\begin{gather}
729 S_1^2 w^4+\left(1728 S_1 S_3-1296 S_1^2 S_2\right) w^2+256 S_3^2 = 0\,,
\label{e:3b}
\\
19683 S_1^3 w^6+\left(69984 S_1^2 S_3-78732 S_1^5\right) w^4+\left(82944 S_1 S_3^2-186624 S_1^3 S_4\right) w^2+32768 S_3^3 = 0\,.
\label{e:3d}
\end{gather}
\ese
Equation~\eqref{e:3d} is cubic in $w^2$, while~\eqref{e:3b} is quadratic in $w^2$. For $S_1\ne0$, the two values determined by~\eqref{e:3b} are
\be
\label{e:wroots}
w_\pm^2=
\frac{8}{27S_1}
\left[
3S_1S_2-4S_3
\pm
\sqrt{3(S_1S_2-2S_3)(3S_1S_2-2S_3)}
\right].
\ee

The two candidate values in~\eqref{e:wroots} admit a simple interpretation. Suppose that the symmetric invariants $S_1,S_2,S_3$ arise from a state $(\rho,s_3,w_0)$ in the symmetric frame. Using the first three relations in~\eqref{e:inversion-v0} to express $S_1,S_2,S_3$ in terms of $(\rho,s_3,w_0)$, 
Equation~\eqref{e:3b} factors as
\be
\frac{81\rho^2s_3^2}{w_0^4}(w^2-w_0^2)\left(9w_0^2w^2-(2\rho+w_0^2)^2\right)=0.
\label{e:w2factorization}
\ee
Thus, the two candidate values are 
\vspace*{-1ex}
\be
 w^2  = w_0^2,\qquad
 w^2  = \frac{(2\rho+w_0^2)^2}{9w_0^2},
\label{e:w2candidates}
\ee
 only the first of which recovers the original value $w_0^2$. 

Returning to the inversion problem, each real nonnegative value among the two candidates $w_\pm^2$ in~\eqref{e:wroots} gives the two signed candidate values $w=\pm[w_\pm^2]^{1/2}$. These values  must still satisfy the remaining equation~\eqref{e:3d}, which is the constraint imposed by the symmetric reduction $v=0$. At this point, a natural question is which of the resulting roots is physically relevant. As we discuss next, two distinct issues arise in this regard.

The first issue is related to an exact residual symmetry.  
A direct computation shows that, when $v=0$, the branch-point quartic $\Delta(k,\@y)$ 
in~\eqref{e:Deltadef} is invariant under the map
\be
(\rho,s_3,w)\;\mapsto\;(\rho,-s_3,-w)\,.
\label{e:residualsymmetry}
\ee
This is simply the statement that interchanging the labels of the two components 
(which reverses the sign of both the polarization imbalance and the counterflow) leaves the 
spectral curve unchanged.
Consequently, the spectral data cannot distinguish between the two sets of values in~\eqref{e:residualsymmetry}.  
Thus, 
any inversion is necessarily subject to this two-fold ambiguity.
This ambiguity accounts for the sign of $w$ associated with a fixed value of $w^2$, and is resolved by any convention fixing the labeling of the two components, e.g.\ $w\ge0$.

The second issue is genuine root selection, and is resolved by physical admissibility
constraints.
Recall that, by definition, $\rho>0$, $s_3\in[-1,1]$ and $w\in\Real$.
Given the symmetric functions $S_1,\dots,S_4$, each candidate root~$w$ of~\eqref{e:3b} determines
$s_3$ and~$\rho$ uniquely through~\eqref{e:s3explicit} and~\eqref{e:rho}.
We therefore adopt the following selection rule: 

\begin{definition}
\label{d:admissible}
We say that a root $w$ of~\eqref{e:3b} is~\textit{admissible} if and only if 
\be
w\in\Real\setminus\{0\}\,,\qquad
16 S_3 + 9 S_1 w^2 \ne 0\,,\qquad
\rho = S_1 w/s_3 > 0\,,\qquad
s_3\in[-1,1]\,,
\label{e:selectionrule}
\ee
where $s_3$ is given by~\eqref{e:s3explicit} and $\rho$ by~\eqref{e:rho}, together with the requirement that the resulting triple $(\rho,s_3,w)$ also 
satisfy the remaining constraint~\eqref{e:3d}.
\end{definition}

\begin{proposition}
\label{p:inversionuniqueness}
If $S_1,\dots,S_4$ arise from a physical state $(\rho,s_3,w_0)$ in the symmetric frame with $w_0s_3\ne0$. the admissible inverse state is unique up to the residual symmetry~\eqref{e:residualsymmetry}.
\end{proposition}

\begin{proof}
Suppose that a signed root corresponding to the second candidate  for $w^2$  in~\eqref{e:w2candidates} is admissible in the sense of Definition~\ref{d:admissible}, and let
\vspace*{-1ex}
\be
\widehat w=\frac{2\rho+w_0^2}{3w_0}.
\ee
[\unskip The opposite sign is accounted for by the residual symmetry~\eqref{e:residualsymmetry}.]\ 
Substituting $\widehat w$ into~\eqref{e:s3explicit} and~\eqref{e:rho} gives
\be
\widehat s_3=\frac{3\rho s_3}{4w_0^2-\rho},
\qquad
\widehat\rho=\frac{(4w_0^2-\rho)(2\rho+w_0^2)}{9w_0^2}.
\label{e:secondbranchreconstruction}
\ee
The admissibility conditions in Definition~\ref{d:admissible} require, in particular, $\widehat\rho>0$ and $|\widehat s_3|\le1$. Since $\rho>0$ and $w_0\ne0$, the first and second conditions imply, respectively 
\be
4w_0^2-\rho>0,\qquad
3\rho|s_3|\le4w_0^2-\rho.
\label{e:secondbranchadmissibility}
\ee

It remains to impose the remaining constraint~\eqref{e:3d}, equivalently the fourth symmetric invariant. Substituting the second of~\eqref{e:secondbranchreconstruction} into the expression for $S_4$ yields
\be
\widehat S_4-S_4=
\frac{(w_0^2-\rho)(\rho+2w_0^2)\big(9\rho^2s_3^2-(\rho+2w_0^2)^2\big)}{48w_0^4}.
\label{e:S4secondbranch}
\ee
Since $\rho>0$ and $w_0\ne0$, compatibility requires either
$w_0^2=\rho$ or $3\rho|s_3|=\rho+2w_0^2.$
In the latter case, combining~\eqref{e:secondbranchadmissibility} with
$3\rho|s_3|=\rho+2w_0^2$ gives
\[
\rho+2w_0^2\le4w_0^2-\rho,
\]
and hence $w_0^2\ge\rho$. On the other hand, the physical bound $|s_3|\le1$ gives
\[
\rho+2w_0^2=3\rho|s_3|\le3\rho,
\]
and hence $w_0^2\le\rho$. Therefore $w_0^2=\rho$. At this value, the two candidates in~\eqref{e:w2candidates} coincide:
 $(2\rho+w_0^2)^2/(9w_0^2)=w_0^2$. 
Thus, the second algebraic branch cannot produce a distinct admissible physical state.
\end{proof}

\section{Modulational stability and instability of plane wave solutions}
\label{s:modulationalstability}

Complex characteristic speeds of the dispersionless system produce exponential growth at leading order in the long-wave limit of the linearized problem~\cite{CourantHilbertII,Lax1957,Whitham1974}, which we will refer to as baseband modulational instability as in~\cite{Baronio2014}. 
Conversely, real characteristic speeds  imply linear stability with respect to sufficiently long-wave perturbations. 
When real characteristic speeds coincide, however, the dispersionless system is degenerate at Whitham order, and their reality alone does not determine the higher-order behavior of the full linearized spectrum.
The availability of an analytical characterization for the Riemann invariants and characteristic speeds of the dispersionless Manakov system allows us to perform a study of baseband modulational stability of the underlying plane wave solutions of the Manakov system, a task that we turn to next.
Importantly, however, note that baseband stability does not rule out the possibility of unstable bands at nonzero perturbation wavenumbers.
The calculation of the full Fourier spectrum will then be carried out in Section~\ref{s:linearization}.

Thanks to the Galilean invariance of the Manakov system, without loss of generality we can again limit ourselves to studying the 
symmetric case $v=0$, i.e., $v_1 = -v_2 = \half w$.
In this case, the characteristic polynomial becomes $P(c,\@y) = p(c,\@z(\@y))$, with
\be
p(c,\@z) = c^4 - (\rho + \half w^2)\,c^2 - \rho w s_3\,c + \tfrac1{16} w^2(w^2 - 4\rho)\,.
\label{e:p4reduced}
\ee
(Note that $p(c,\@z)$ does not contain $v$, consistent with the assumption $v=0$.)
The key question now is whether the roots $c_1,\dots,c_4$ of $p(c,\@z)$ are real or complex.
The reality or complexity of the roots of $p(c,\@z)$ determine whether sufficiently long-wave perturbations are stable or unstable. 
Next we discuss the case of distinct roots.
The case of repeated roots will be considered separately below.

The expression~\eqref{e:p4reduced} for $p(c,\@z)$ is simple enough that it allows for an analytical study.
First of all, note that, in the reduction to a single component, i.e., when $\rho_2=0$ or $\rho_1=0$ (which is obtained when $s_3 = \pm1$, respectively), the results agree with those for the defocusing NLS equation.
Indeed, for $s_3=\pm1$, the characteristic speeds are
\[
c=-\tfrac12s_3w \quad\text{(with multiplicity two)},
\qquad
c=\tfrac12s_3w\pm\sqrt\rho.
\]
Similarly, it is easy to check that there is no baseband instability when no counterflow is present i.e., when $w=0$.
Indeed, in this case one obtains $c=0$ (which is again a double root) and $c=\pm\sqrt\rho$,
which also coincides with the limit $w=0$ of the values obtained when $s_3=\pm1$.

To analyze the general case, we can use the results of \cite{Rees}, conveniently reformulated for our purposes in Lemma~2.3 of~\cite{Stevic}:

\begin{lemma}[Rees, 1922]
\label{l:Stevic}
Consider the quartic polynomial
\bse
\begin{gather}
p_4(y)=y^4+by^3+cy^2+dy+e,
\label{e:p4def}
\\
\noalign{\noindent and let}
    \D = \tfrac1{27}(4\D_0^3-\D_1^2),\qquad 
    \D_0 = c^2-3bd+12e,\qquad
    \D_1 = 2c^3-9bcd+27b^2e+27d^2-72ce,
\\
    \P = 8c-3b^2,\qquad 
    \Q = b^3+8d-4bc,\qquad 
    \R = 64e-16c^2+16b^2c-16bd-3b^4\,.
\end{gather}
\ese
\begin{itemize}
\item[(a)] 
If $\D<0$, then two zeros of $p_4$ are real and distinct, while two are complex conjugates.
\item[(b)] 
If $\D>0$, then all the zeros of $p_4$ are real or none is.  Specifically:
    \begin{itemize}
    \item[(b1)] 
    if $\P<0$ and $\R<0$, then all four zeros of $p_4$ are real and different;
    \item[(b2)] 
    if $\P>0$ or $\R>0$, then there are two pairs of complex conjugate zeros of $p_4$.
    \end{itemize}
\item[(c)] 
If $\D=0$, then and only then $p_4$ has a multiple zero. The following cases can occur:
    \begin{itemize}
    \item[(c1)] 
    if $\P<0$, $\R<0$ and $\D_0\ne 0$, then two zeros of $p_4$ are real and equal and two are real and simple;
    \item[(c2)] 
    if $\R>0$ or ($\P>0$ and ($\R\ne 0$ or $\Q\ne 0$)), then two zeros of $p_4$ are real and equal and two are complex conjugate;
    \item[(c3)] 
    if $\D_0=0$ and $\R\ne 0$, there is a triple zero of $p_4$ and one simple, all real;
    \item[(c4)] if $\R=0$, then
        \begin{itemize}
        \item[(i)] 
        if $\P<0$ there are two double real zeros of $p_4$;
        \item[(ii)] 
        if $\P>0$ and $\Q=0$ there are two double complex conjugate zeros of $p_4$;
        \item[(iii)] 
        if $\D_0=0$, then all four zeros of $p_4$ are real and equal to $-b/4$.
        \end{itemize}
    \end{itemize}
\end{itemize}
\end{lemma}

To avoid confusion, we should point out that, for simplicity, in Lemma~\ref{l:Stevic} 
the names of the coefficients of $p_4(y)$ follow the notation of \cite{Stevic},
so the quantity $c$ there should not be confused with the variable with the same name used earlier in this work.
Conversely, the variable~$y$ in Lemma~\ref{l:Stevic} should be identified with the parameter $c$ in \eqref{e:p4reduced}.
[Note also that $\D$ is the discriminant of $p_4(y)$.]

Specifically, when the polynomial $p_4(y)$ in~\eqref{e:p4def} is identified with $p(c,\@z)$ in~\eqref{e:p4reduced}, the relevant quantities in Lemma~\ref{l:Stevic} are as follows:
\bse
\begin{gather}
\D = (1 - s_3^2) w^2 \rho^2 \big(27 \rho^2 s_3^2 w^2 + (w^2-4 \rho)(\rho + 2w^2)^2 \big)\,,
\label{e:Discriminant2}
\\
\P=-4(w^2+2\rho),
\qquad
\R=-16\rho(\rho+2w^2),
\qquad
\D_0=(w^2-\rho)^2.
\label{e:PD_Delta0}
\end{gather}
\ese

\begin{proposition}
\label{p:branch-char-discriminants}
On the generic set $\rho>0$, $w\ne0$, and $|s_3|<1$, two branch points collide if and only if two characteristic speeds coincide.
\end{proposition}

\begin{proof}
Let $\Delta(k,\@y)$ denote the branch-point quartic in~\eqref{e:Deltadef}, and let $\D$ denote the discriminant of the characteristic polynomial~\eqref{e:p4reduced}. In the symmetric frame $v=0$,
\be
\operatorname{disc}_k\Delta
=
4096\,\mathcal B^2\,\D,
\label{e:branch-char-discriminant-identity}
\ee
where
\be
\mathcal B = 27\rho^2s_3^2w^2+(w^2-4\rho)(\rho+2w^2)^2.
\label{e:branch-char-factor}
\ee
A direct calculation of the discriminant of the quartic~\eqref{e:Deltadef} with respect to $k$ gives
\be
\operatorname{disc}_k\Delta = 4096\rho^2w^2(1-s_3^2)\mathcal B^3.
\ee
On the other hand,~\eqref{e:Discriminant2} gives
\be
\D=\rho^2w^2(1-s_3^2)\mathcal B.
\ee
Combining the two identities yields~\eqref{e:branch-char-discriminant-identity}.
\end{proof}

In what follows, we study the sign of these quantities depending on the relative values of $\rho$, $s_3$ and $w$ and use this analysis to identify the regions where the characteristic speeds are all real or where some of them are complex.
Note that the locus where $\D=0$ is simply the boundary between these regions.  
We will first limit ourselves to studying cases~(a) and~(b) in Lemma~\ref{l:Stevic},
ignoring case~(c).  
Afterwards, in Lemma~\ref{l:Delta=0} we will characterize the characteristic speeds when $\D=0$.

Note first that: (i) $\P$ and $\R$ are independent of $s_3$, and (ii) $\P$ and $\R$ are always negative
(except in the trivial case $\rho=0$).  
So the analysis reduces simply to the study of the discriminant. 
Namely:
\vspace*{-1ex}
\begin{enumerate}
\advance\itemsep-4pt
\item[(a)] 
If $\D>0$, there are four real characteristic speeds;
\item[(b)]
Conversely, if $\D<0$, there are two real speeds and two complex conjugate speeds.
\end{enumerate}

\begin{figure}[b!]
\centering
\includegraphics[width=0.85\textwidth]{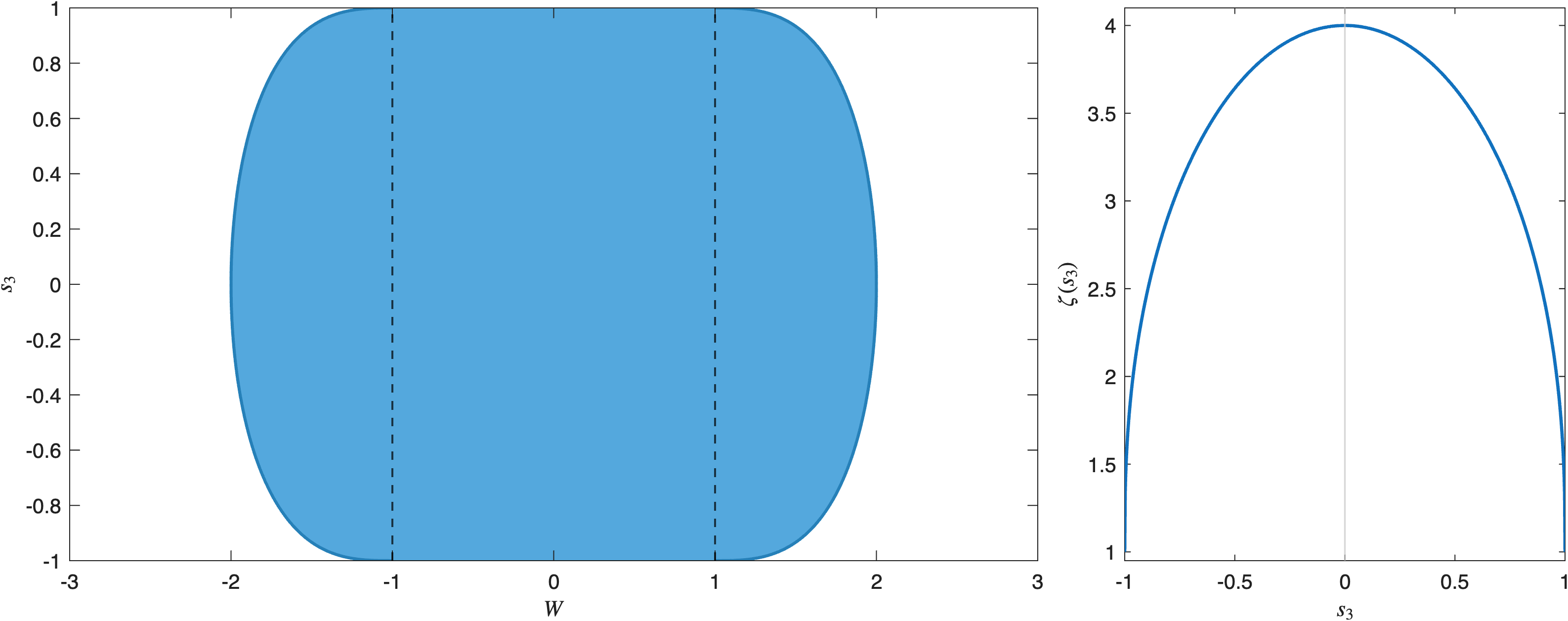}
\caption{Left panel: The region (in blue) of the $Ws_3$-plane (with $W$ as the horizontal axis and~$s_3$ in the vertical axis) for which the plane wave solutions of the defocusing Manakov system are baseband modulationally unstable at leading order in the long-wave limit.  
Right panel: the function $\zeta(s_3)$ (vertical axis) defined in equation~\eqref{e:y(s)} as a function of $s_3$ (horizontal axis).}
\label{f:Ws3instabilitydomain}
\end{figure}

Next, note that one can obtain a complete classification by observing that $\rho$ can be completely rescaled away, as one could expect owing to the scaling invariance of the Manakov system.
Indeed, the change of variables $c = \sqrt\rho \,C$ and $w = \sqrt\rho\,W$ yields
$p(c,\@z) = \rho^2 g(C,W,s_3)$ and $\D = \rho^6 (1-s_3^2) W^2 d(W,s_3)$, with
\bse
\begin{gather}
g(C,W,s_3) = C^4-\half C^2(W^2+2) - s_3 C W + \tfrac1{16} W^2 (W^2-4)\,,\\
d(W,s_3) = 4W^6-12W^4+3(9s_3^2-5)W^2-4 = (W^2-4)(2W^2+1)^2+27s_3^2W^2\,.
\end{gather}
\ese
The domain of the $Ws_3$-plane where $d(W,s_3)<0$ 
(shown as the blue region in Fig.~\ref{f:Ws3instabilitydomain})
identifies the parameter values where the plane wave solutions of the defocusing Manakov system are  modulationally unstable at leading order in the long-wave limit.  
Specifically, we have:

\begin{lemma}
\label{l:unstable}
The counterpropagating plane waves are baseband unstable when
\[
s_3\in(-1,1),
\qquad
0<|w|<2\sqrt\rho,
\qquad
|s_3|<\min\bigg\{\sqrt{f\left(w/\sqrt\rho\right)},1\bigg\},
\]
where
\be
f(W)=\frac{(4-W^2)(2W^2+1)^2}{27W^2}.
\label{e:fWdef}
\ee
The zero-counterflow state $w=0$ is baseband stable.
\end{lemma}

\begin{proof}
The threshold for instability is given by the curve $d(W,s_3)=0$, which yields 
$s_3 = \pm \sqrt{f(W)}$, with $f(W)$ as above.
The instability threshold in the original variables is easily obtained by simply recalling that $W = w/\sqrt\rho$.
\end{proof}

We point out that the baseband criterion in Lemma~\ref{l:unstable} is equivalent to the condition obtained in~\cite{Baronio2014} by Baronio et al. In their notation, the defocusing baseband-instability condition is
\begin{equation*}
(a_1^2+a_2^2)^3
-12(a_1^4-7a_1^2a_2^2+a_2^4)q^2
+48(a_1^2+a_2^2)q^4
-64q^6>0.
\end{equation*}
Using the change of variables  
$\rho=a_1^2+a_2^2$,
$s_3=(a_1^2-a_2^2)/\rho$ and 
$q=\half w=\half\sqrt{\rho}\,W$, 
one has
$a_1^4-7a_1^2a_2^2+a_2^4=\frac14\rho^2(9s_3^2-5)$,
and the above inequality becomes
$\frac14\rho^3\,d(W,s_3)<0$.
Hence, the condition of~\cite{Baronio2014} is exactly $d(W,s_3)<0$, the criterion obtained here from the characteristic-speed discriminant of the dispersionless Manakov system.

Also note that $f(W)>1$ for $0<|W|<1$, with $f(\pm1)=1$. 
Thus, for all $W\in(-1,1)$ [i.e., in the portion in Fig.~\ref{f:Ws3instabilitydomain} between the two vertical dashed lines], all values of $s_3\in(-1,1)$ correspond to baseband-unstable solutions.
Indeed, an expansion of the roots $C_1,\dots,C_4$ around $W=0$ yields
$C_{1,2} = \pm 1+\half Ws_3 + O(W^2)$ and 
$C_{3,4} = \half ( - s_3 \pm i \sqrt{1-s_3^2} ) W+O(W^2)$.
Conversely, an expansion around $W=\infty$ 
yields
$C_{1,2} = -\half W  \pm\sqrt{(1-s_3)/2} + O(1/W)$ and 
$C_{3,4} = \half W \pm\sqrt{(1+s_3)/2} + O(1/W)$
(with a possible reordering of the roots),
showing that no asymptotic baseband instability exists.
Indeed, it is easy to see from~\eqref{e:fWdef} that no baseband instability is present when $|W|>2$,
since $f(W)<0$ for $|W|>2$. 

For a Fourier mode $\e^{\ii\xi(x-ct)}$, the temporal growth rate is $|\xi\,\Im c|$. 
Since $c=\sqrt\rho\,C$, we will therefore call
\begin{equation}
I(W,s_3):=|\Im C|
\end{equation}
the dimensionless baseband growth coefficient.

\begin{lemma}
\label{l:growthrate}
The maximum dimensionless baseband growth coefficient over the unstable region is attained at
\[
s_3=0,
\qquad
|W|=\sqrt{3/2},
\]
and equals $I_{\max}=1/(2\sqrt2)$. For each fixed $W$, $I(W,s_3)$ is strictly decreasing as a function of $|s_3|$ throughout the unstable region.
\end{lemma}

\begin{proof}
First, note that, since the coefficients of $g(C,W,s_3)$ in~\eqref{e:p4reduced} are real, complex roots occur in conjugate pairs.  
Moreover, the classical Ferrari resolvent-cubic factorization for a depressed quartic implies that 
the polynomial factors over $\Real$ as
\be
g(C,W,s_3) = (C^2+p\,C+r_+)(C^2-p\,C+r_-)\,,
\label{e:quarticfactorization}
\ee
for suitable real quantities $p,r_+,r_-$.
Explicitly, comparing the coefficients in~\eqref{e:quarticfactorization} with those in~\eqref{e:p4reduced} yields
\be
r_+ + r_- = p^2 - \tfrac12(W^2+2)\,,\qquad
p\,(r_--r_+) = -\,s_3 W\,,\qquad
r_+\,r_- = \tfrac1{16}W^2(W^2-4)\,.
\label{e:factorizationconditions}
\ee
Further, eliminating $r_\pm$ from~\eqref{e:factorizationconditions} shows that $P := p^2$ must be a
root of the resolvent cubic
\be
h(P;W,s_3) := P^3 - (W^2+2)\,P^2 + (2W^2+1)\,P - W^2 s_3^2\,.
\label{e:resolventcubic}
\ee
Next, writing $\zeta:=W^2$ and introducing the quantity
\be
R(P):=P^2-(\zeta+2)P+(2\zeta+1)\,, 
\label{e:Rdef}
\ee
the resolvent cubic takes the compact form 
\be
h(P;W,s_3)=P\,R(P)-\zeta\,s_3^2, 
\ee
so that, on its root locus, one has
\be
\zeta\,s_3^2 = P\,R(P)\,. 
\label{e:compactidentity}
\ee
In particular, one has $R(P)\ge0$ automatically whenever $P\ge0$, as it must be in order for $r_\pm$ to be real.
The baseband-unstable configurations are precisely those for which one of the two quadratic factors
in~\eqref{e:quarticfactorization} has negative discriminant, and the corresponding baseband growth coefficient satisfies $I^2 = r_+ - \tfrac14 P$ for that factor.

Finally, solving~\eqref{e:factorizationconditions} for $r_+$ and using~\eqref{e:resolventcubic} to eliminate
$s_3$ in favor of~$P$, one obtains the key representation
\be
I^2 = I_2(P;W) := \tfrac14\big(P - (W^2+2)\big) + \tfrac12\sqrt{P^2-(W^2+2)P+(2W^2+1)}\,,
\label{e:I2ofP}
\ee
in which the polarization imbalance $s_3$ no longer appears explicitly.
In other words, the entire
$s_3$-dependence of the baseband growth coefficient is carried by the single scalar~$P$.
Throughout the unstable region, $P$ is the unique real root of~\eqref{e:resolventcubic}. 
And since $I_2(P;W)$ depends on $s_3$ only through $P$, the chain rule gives
\be
\frac{\partial (I^2)}{\partial (s_3^2)} = \frac{dI_2}{dP}\,\frac{\partial P}{\partial (s_3^2)}\,.
\label{e:chainrule}
\ee
Thus, it suffices to pin down the sign of each of the two factors on the right-hand side.
We do so next in two separate steps.

\textit{Step 1: $P$ is strictly increasing with respect to $s_3^2$.}
Throughout the interior of the unstable set, the real root $P$ of~\eqref{e:resolventcubic} is unique and simple. Indeed, with $\zeta=W^2$, the discriminant of the resolvent cubic is
\[
\operatorname{disc}_P h
=
\zeta(1-s_3^2)d(W,s_3)<0,
\]
since $\zeta>0$, $|s_3|<1$, and $d(W,s_3)<0$ in the unstable region. Thus $h$ has exactly one real root. Since $h(P;W,s_3)$ is a monic cubic, it crosses from negative to positive at this root, and therefore $h'(P)>0$.
Implicit differentiation therefore gives
\be
\frac{\partial P}{\partial(s_3^2)}=\frac{W^2}{h'(P)}>0,
\label{e:dPds3}
\ee
as anticipated.

\textit{Step 2: $I_2(\,\cdot\,;W)$ is strictly decreasing on the unstable set.}
Letting $\zeta=W^2$ and $R(P)=P^2-(\zeta+2)P+(2\zeta+1)$ as before, we have 
\be
\frac{dI_2}{dP} = \frac{2P-(\zeta+2)+\sqrt{R(P)}}{4\sqrt{R(P)}}. 
\label{e:dPsidP}
\ee
It is straightforward to see that inequality $dI_2/dP<0$ is equivalent to requiring
\begin{gather}
\zeta+2-2P>0\,, 
\qquad
G(P;\zeta):=3P^2-3(\zeta+2)P+(\zeta^2+2\zeta+3)>0. 
\label{e:Gdef}
\end{gather}
Next, evaluating $h(P;W,s_3)$ at $P=0$ and $P=\zeta$, with $\zeta=W^2$, gives
\[
h(0;W,s_3)=-\zeta s_3^2\le0,
\qquad
h(\zeta;W,s_3)=\zeta(1-s_3^2)>0.
\]
Since $P$ is the unique real root of~\eqref{e:resolventcubic}, it follows that
\[
0\le P<\zeta.
\]
Thus, since $\zeta+2-P>0$, the condition $I^2>0$ in~\eqref{e:I2ofP} is equivalent to 
\be
4R(P)-(\zeta+2-P)^2=(P-\zeta)(3P+\zeta-4)>0. 
\label{e:unstablecharacterization}
\ee
The first factor is negative, and hence
$P<\zeta$ and $P<\tfrac13(4-\zeta)$. 
These bounds also imply $\zeta+2-2P>0$. The smaller root of $G(P;\zeta)$ is 
\be
P_-=\tfrac12(\zeta+2)-\tfrac16\sqrt{3\zeta(4-\zeta)}. 
\ee
Next, note that, for $0<\zeta<1$, 
\bse
\be
(1-\tfrac12 \zeta)^2-\tfrac1{12} \zeta(4-\zeta) = \tfrac13(\zeta-1)(\zeta-3)>0, 
\label{e:twoinequalities-a}
\ee
so $P_->\zeta>P$. 
Moreover, for $1<\zeta<4$, 
\be
(\tfrac56\zeta-\tfrac13)^2-\tfrac1{12}\zeta(4-\zeta) = \tfrac19(\zeta-1)(7\zeta-1)>0, 
\label{e:twoinequalities-b}
\ee
\ese
so $P_->(4-\zeta)/3>P$. 
Equality holds at $\zeta=1$. 
Thus $P<P_-$ throughout the unstable set, $G(P;\zeta)>0$, and $dI_2/dP<0$, 
as anticipated.

In conclusion, combining Steps~1 and~2, 
at every point of the unstable region
one has
\be
\frac{\partial (I^2)}{\partial (s_3^2)} = \frac{dI_2}{dP}\,\frac{\partial P}{\partial (s_3^2)} < 0.
\label{e:globalmonotonicity}
\ee
Hence, for each fixed~$W$, the baseband growth coefficient $I(W,s_3)$ is a strictly decreasing function of $|s_3|$,
and its global maximum over $s_3\in(-1,1)$ is therefore attained when $s_3=0$.

The last remaining task is to find the maximum of the growth coefficient along the line $s_3=0$.
There, \eqref{e:p4reduced} is biquadratic and the rescaled characteristic speeds are simply
\be
C_{1,2} = \pm\tfrac12\sqrt{W^2 - 2\sqrt{2W^2+1} + 2}\,,\qquad
C_{3,4} = \pm\tfrac12\sqrt{W^2 + 2\sqrt{2W^2+1} + 2}\,,
\label{e:s3=0speeds}
\ee
up to a reordering.
Of these, $C_{3,4}$ are always real, while $C_{1,2}$ are purely imaginary precisely when
$W^2+2 < 2\sqrt{2W^2+1}$, i.e., for $0<W^2<4$.
On that interval,
\be
I(W,0)^2 = \tfrac14\big(2\sqrt{2W^2+1} - W^2 - 2\big)\,,
\label{e:Ionaxis}
\ee
The function $I(W,0)^2$ is even in $W$, and its derivative with respect to $W^2$ vanishes at $W^2=3/2$, where the second derivative is negative. Hence the two maximizers are $W=\pm\sqrt{3/2}$, and~\eqref{e:Ionaxis} gives $I_{\max}^2=1/8$, or $I_{\max}=1/(2\sqrt2)$.

Returning to the original variables, since $c=\sqrt{\rho}\,C$ and $w=\sqrt{\rho}\,W$, the corresponding dimensional maximum is
$\max|\Im c|=\sqrt{\tfrac18\rho}$,
which is attained at $s_3=0$ and $|w|=\sqrt{\tfrac32\rho}$.
\end{proof}

\begin{remark}
The strict monotonicity~\eqref{e:globalmonotonicity} can be checked independently at $s_3=0$ 
by a direct second-derivative computation.  
Tracking a single root $C(s_3)$ of $g(C,W,s_3)=0$ by implicit differentiation, and noting that the
transformation $(C,s_3)\mapsto(-C,-s_3)$ leaves~$g$ invariant, one has $I(W,s_3)=I(W,-s_3)$, so 
that $s_3=0$ is automatically a critical point; consistently, 
$C'(0) = -W/\big(2\sqrt{2W^2+1}\big)\in\Real$, so the root moves purely horizontally to first order 
and $(I^2)'(0)=0$.  
Carrying the expansion to second order yields
\be
\left.\frac{d^2(I^2)}{ds_3^2}\right|_{s_3=0}
  = \frac{W^2\big[\,2W^2+1-(W^2+2)\sqrt{2W^2+1}\,\big]}{2\,(2W^2+1)^2}\,,
\label{e:d2I2}
\ee
whose sign is that of the bracket $N := 2W^2+1-(W^2+2)\sqrt{2W^2+1}$.
Substituting $r=\sqrt{2W^2+1}\ge1$ (so that $W^2 = (r^2-1)/2$) gives 
$N = -\tfrac{r}{2}\big[(r-1)^2+2\big]<0$ for all $r\ge1$, in agreement 
with~\eqref{e:globalmonotonicity}.
\end{remark}

 \begin{figure}[t!]
\centering
\includegraphics[width=0.95\textwidth]{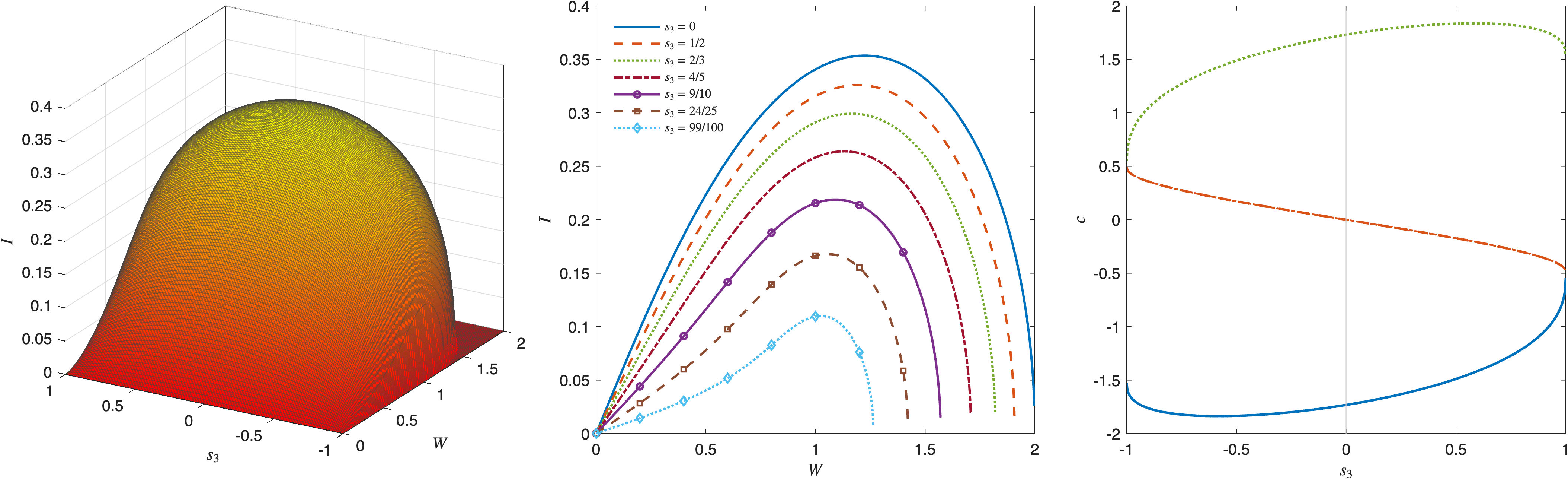}
\caption{Left panel: The baseband growth coefficient $I$ (vertical axis) as a function of $W$ and $s_3$, evaluated by numerically computing the characteristic speeds. Center panel: The baseband growth coefficient $I$ as a function of $W$ (horizontal axis) for $s_3=0$ (blue), 1/2 (orange), 2/3 (green), 4/5 (red), 9/10 (purple), 24/25 (brown) and 99/100 (cyan). Right panel: The characteristic speeds (vertical axis) as a function of $s_3$ (horizontal axis) in the non-strictly hyperbolic sector where two of them (shown by the orange curve) coincide.}
\label{f:growthrate}
\end{figure}

Figure~\ref{f:growthrate} shows the baseband growth coefficient as a function of both $W$ and $s_3$ (left panel) and as a function of $W$ for various values of $s_3$ (center panel).

It remains to characterize what happens when the discriminant $\D$ in \eqref{e:Discriminant2} vanishes.
As mentioned in section~\ref{s:modulationalstability}, two characteristic speeds coincide if and only if $\D=0$.
When this happens, the dispersionless Manakov system is not strictly hyperbolic (cf.~\cite{ElHoefer2016}).
Next we characterize the characteristic speeds on this degenerate locus.  We will show below that in this case the speeds  are all real, so the leading Whitham growth coefficient vanishes. Since at least one characteristic speed is repeated, however, this does not by itself imply spectral stability for nonzero perturbation wavenumber.

For $|s_3|<1$, the condition $\D=0$ contains the zero-counterflow branch $W=0$, for which baseband instability is not present. The remaining branches are determined by $d(W,s_3)=0$. 
Since $d(W,s_3)$ is even in $W$, it is convenient to set $\zeta=W^2$, which reduces this condition to a cubic equation in $\zeta$. Of its three roots, two are complex and therefore nonphysical, while the remaining real root is 
\be
\zeta(s) = \tfrac32 \sqrt[3]{(1-s_3)^2(1+s_3)} + \tfrac32 \sqrt[3]{(1-s_3)(1+s_3)^2}+1\,. 
\label{e:y(s)}
\ee
A plot of $\zeta(s)$ is shown in the right panel of Fig.~\ref{f:Ws3instabilitydomain}. 
Inserting \eqref{e:y(s)} into~\eqref{e:p4reduced} allows one to numerically find the characteristic speeds.
The resulting values, all real, are shown in the right panel of Fig.~\ref{f:growthrate}.
Of course a numerical evaluation does not rigorously prove the reality. 
On the other hand, we have the following:

\begin{lemma}
\label{l:Delta=0}
When the discriminant~$\D$ in~\eqref{e:Discriminant2} vanishes and $w^2\ne\rho$, one characteristic speed is real and double and the other two are real and simple. 
When $\D=0$ and $w^2=\rho$, one characteristic speed is real and triple and the remaining speed is real and simple.
\end{lemma}

\begin{proof}
Under the hypothesis $\D=0$, case~(c) of Lemma~\ref{l:Stevic} applies. Equation~\eqref{e:PD_Delta0} gives $\P<0$ and $\R<0$. If $w^2\ne\rho$, then $\D_0=(w^2-\rho)^2>0$, and case~(c1) gives one real double root and two real simple roots. If $w^2=\rho$, then $\D_0=0$ and $\R\ne0$, so case~(c3) gives one real triple root and one real simple root.
\end{proof}

The triple-root case in Lemma~\ref{l:Delta=0} occurs only at the scalar endpoints. 
Indeed, $w^2=\rho$ is equivalent to $W^2=1$, and
\[
d(\pm1,s_3)=27(s_3^2-1).
\]
Hence $\D=0$ and $w^2=\rho$ require $|s_3|=1$. In the genuinely two-component region $|s_3|<1$, the degenerate locus therefore consists only of one real double characteristic speed and two real simple characteristic speeds. Since $\mathcal K=W/\sqrt2$ in section~\ref{s:linearization}, the exceptional value $W^2=1$ corresponds to $\mathcal K^2=\half$; this is the same scalar-endpoint degeneracy excluded in Proposition~\ref{p:degenerate-longwave-instability} .

Thus, on the locus $\D=0$, the characteristic speeds are all real but are not all distinct. Accordingly, the dispersionless growth coefficient $I$ vanishes there, and the locus is marginal at Whitham order. The dispersionless calculation alone does not determine whether instability persists at higher order as the perturbation wavenumber tends to zero; this question requires the full linearized dispersion relation considered in section~\ref{s:linearization}.

It is worth pausing briefly to interpret the above results physically.
The mechanism identified by the analysis is unambiguous: the instability is driven entirely by the
counterflow~$w$ between the two components, and is suppressed by the polarization imbalance~$s_3$.
Indeed, no baseband instability is present when $w=0$ no matter what the polarization is, nor when 
$s_3=\pm1$ (i.e., in the scalar reduction) no matter what the counterflow is; and, by 
Lemma~\ref{l:growthrate}, at fixed counterflow the baseband growth coefficient decreases monotonically as the imbalance increases.
The instability is therefore a genuinely two-component effect: it requires both components to be
populated and to move relative to one another. 

This counterflow-driven mechanism is closely related to the counter-superflow instability studied for miscible two-component Bose--Einstein condensates. 
Refs.~\cite{PhysRevLett.106.065302,Law2001} demonstrated that two interacting condensates moving through one another become dynamically unstable once their relative velocity exceeds a critical value. The corresponding Bogoliubov instability and its nonlinear development were subsequently studied in detail in~\cite{TakeuchiIshinoTsubota2010,IshinoTsubotaTakeuchi2011}, where the counter-superflow instability was shown to generate momentum transfer between the two components and, in higher dimensions, vortex nucleation and binary quantum turbulence. The Manakov system considered here corresponds to the equal-coupling, miscibility-threshold limit and provides an integrable setting in which the counterflow instability can be characterized explicitly through the Whitham characteristic speeds and the full finite-wavenumber discriminant.
 
Complex characteristic speeds occur only for $0<|w|<2\sqrt\rho$ and sufficiently small imbalance. When $|w|>2\sqrt\rho$, all characteristic speeds are real, so sufficiently long-wave perturbations are stable. This result does not imply stability at every perturbation wavenumber, however. 
The direct calculation in section~\ref{s:linearization} shows that finite-wavenumber unstable bands can exist outside the baseband region, in agreement with the intermediate-wavelength instability described in~\cite{JNLS2000v10p291}. 
Large counterflow therefore does not decouple the two components into independent scalar equations.
This picture is consistent with, and complements, what is known in the closely related setting of
nonlinear polarization waves in two-component
condensates~\cite{KKLP2014,LPK2013,Kamchatnov2014,CongyKamchatnovPavloff2016},
where the polarization (counter-phase) mode is precisely the one
carrying the counterflow degree of freedom, and where the polarization sound speed --- unlike the
density sound speed --- can become imaginary.
It is also the natural point of contact with the direct linearized analysis 
of~\cite{JNLS2000v10p291}, which we revisit in section~\ref{s:linearization}, and where the
long-wavelength limit of the unstable band is shown to reproduce exactly the threshold obtained
here.
Experimental evidence for both baseband and passband polarization modulation instability in a defocusing Manakov fiber system was reported by Frisquet et al.~\cite{Frisquet2015}.

Finally, we note that in the recent experimental and theoretical study~\cite{PRL2025v135p113401} 
of a related, non-integrable two-component system, the two components were not permitted to
counterpropagate, so the mechanism described above could not operate; understanding to what extent
it survives away from the integrable, miscibility-threshold point is an interesting open question.
\section{Stability analysis via direct linearization of the Manakov system}
\label{s:linearization}

The analysis of section~\ref{s:modulationalstability} describes the limit of vanishing perturbation wavenumber. 
As mentioned earlier, Forest et al.~\cite{JNLS2000v10p291} derived the full linearized dispersion relation for counterpropagating plane waves and analyzed both long-wave and intermediate-wavelength instability. Here we write that dispersion relation in discriminant form over the full $(\mathcal K,\Xi,s_3)$ parameter space and compare its $\Xi\to0$ limit with section~\ref{s:modulationalstability}.

We again set $\eps=1$ in~\eqref{e:Manakov}.  We now also rescale time by letting $t=2\~t$. 
After dropping tildes, we obtain the Manakov system in the form used in this section as
\be
\ii\@q_t+\@q_{xx}-2(\@q^\dagger\@q)\@q=0.
\label{e:ManakovISTnormalization}
\ee
We begin by recalling that an exact plane wave solution of  equation~\eqref{e:ManakovISTnormalization} is given by
\be
\@q_o(x,t) = ( \,q_{o,1} \,,\, q_{o,2}\,)^\t\,,\qquad
q_{o,1}(x,t) = a_1\,\e^{\ii(\kappa x-\omega_o t)}\,,\qquad
q_{o,2}(x,t) = a_2\,\e^{-\ii(\kappa x+\omega_o t)}\,, 
\label{e:qOIC}
\ee
with $\omega_o(\kappa,\@a) = \kappa^2 + 2\|\@a\|^2$,
where $\@a = (a_1,a_2)^\t$
and where again we used Galilean invariance to set the mean wavenumber to zero
as well as the phase invariance to set $a_1$ and $a_2$ to be real and non-negative without loss of generality,
so that $w=v_1-v_2=2\kappa$, or equivalently $\kappa=w/2$. 
Similarly to \cite{JNLS2000v10p291}, we then look for perturbed solutions in the form
\be
\label{e:perturbationzansatz}
\@q(x,t) = (\, q_1 \,,\, q_2\,)^\t,\qquad 
q_j(x,t) = q_{o,j}(x,t)\,\big( 1 + \delta\,u_j(x,t)\,\big)\,,\qquad j = 1,2,
\ee
where the parameter $0<\delta\ll1$ is the perturbation strength.
Substituting~\eqref{e:perturbationzansatz} into~\eqref{e:ManakovISTnormalization} and neglecting terms at $O(\delta^2)$ and higher then yields the linearized Manakov system
for $\@u(x,t) = ( u_1,u_2)^\t$ as
\vspace*{-1ex}
\be
\ii\@u_t+2\ii \kappa\sigma_3\@u_x+\@u_{xx}-2A(\@u+\@u^*)=0, 
\qquad
\sigma_3=\diag(1,-1),
\qquad
A=\begin{pmatrix}a_1^2&a_2^2\\a_1^2&a_2^2\end{pmatrix}.
\label{e:linearizedManakov}
\ee
where the asterisk denotes complex conjugation.
Since~\eqref{e:linearizedManakov} is a linear system of PDEs, without loss of generality we can look for
solutions in the form of pure Fourier modes as
\be
\@u(x,t) = \@f\,\e^{\ii\xi(x - \omega t)} + \@g^*\,\e^{-\ii\xi(x-\omega^*t)}\,.
\label{e:perturbationFourier}
\ee
Here, $\omega_o$ in~\eqref{e:qOIC} is the temporal frequency of the background plane wave, whereas $\omega$ in~\eqref{e:perturbationFourier} is the phase velocity of the perturbation mode; the corresponding temporal frequency is $\xi\omega$. 
Substituting~\eqref{e:perturbationFourier} into~\eqref{e:linearizedManakov} and collecting the coefficients of 
$\e^{\pm\ii\xi(x - \omega t)}$ and $\e^{\pm\ii\xi(x - \omega^*t)}$ then yields the homogeneous linear system of algebraic equations
\bse
\be
M_4\,\@x_4 = \@0,\qquad
M_4 = \@1_4\@a_4^\t + D_4\,,\qquad
\ee
for the unknown vector $\@x_4 = (f_1,g_1,f_2,g_2)^\t$,
with $\@1_4 = (1,1,1,1)^t$,
\begin{equation}
\begin{alignedat}{2}
\@a_4 &= (a_1^2,a_1^2,a_2^2,a_2^2)^t,\qquad
& D_4 &= \diag(d_1,d_2,d_3,d_4),
\\
d_1 &= \half\xi(2\kappa+\xi-\omega),
& d_2 &= \half\xi(-2\kappa+\xi+\omega),
\\
d_3 &= \half\xi(-2\kappa+\xi-\omega),
& d_4 &= \half\xi(2\kappa+\xi+\omega).
\end{alignedat}
\end{equation}
\ese
In order for nontrivial solutions to exist, one needs $\det M_4 = 0$, which yields the linearized dispersion relation $\omega = \omega(\xi)$.
It is a tedious but relatively straightforward calculation to show that 
\be
\det M_4 = d_1 d_2 d_3 d_4 \big( 1 + a_1^2 C_1 + a_2^2 C_2 \big)\,,\qquad
C_1 = \frac1{d_1} + \frac1{d_2}\,,
\qquad
C_2 = \frac1{d_3} + \frac1{d_4}\,.
\ee
Clearing denominators then yields the following implicit equation for the linearized dispersion relation:
\be
P_{\rm lin}(\kappa,\xi,\omega)=0, 
\qquad
P_{\rm lin}(\kappa,\xi,\omega)=
\left((\omega-2\kappa)^2-4a_1^2-\xi^2\right)
\left((\omega+2\kappa)^2-4a_2^2-\xi^2\right)
-16a_1^2a_2^2\,. 
\label{e:linearizeddispersionrelation}
\ee
Equation~\eqref{e:linearizeddispersionrelation} 
(modulo slightly different but equivalent normalizations)
was first obtained in \cite{JNLS2000v10p291}. 
There, the authors then proceeded to study it in various asymptotic limits, since the corresponding polynomial defies direct factorization.
On the other hand, we next show that, despite this fact, it is still possible to obtain a complete  
characterization of the stability and instability regimes.
To this end, it is convenient to again employ symmetrized and rescaled variables by introducing the quantities
\be
a_1^2 = \half \rho(1 + s_3)\,,\qquad
a_2^2 = \half \rho (1 - s_3)\,,\qquad 
\omega = \sqrt{2\rho}\,\Omega\,,\qquad
\kappa = \sqrt{\rho/2}\,\mathcal K\,, \qquad
\xi = \sqrt{2\rho}\,\Xi\,.
\ee
Then one obtains
\bse
\begin{gather}
P_{\rm lin}(\kappa,\xi,\omega)=4\rho^2\widetilde P_{\rm lin}(\mathcal K,\Xi,\Omega,s_3), 
\\
\widetilde P_{\rm lin}(\mathcal K,\Xi,\Omega,s_3)=
\left((\Omega-\mathcal K)^2-\Xi^2-1-s_3\right)
\left((\Omega+\mathcal K)^2-\Xi^2-1+s_3\right)
-(1-s_3^2)\,. 
\label{e:scaledlinearizedquartic}
\end{gather}
\ese
Note how the total density $\rho$ has completely disappeared from $\widetilde P_{\rm lin}(\mathcal K,\Xi,\Omega,s_3)$. Therefore, we can now characterize the roots $\Omega_1,\dots,\Omega_4$ of $\widetilde P_{\rm lin}(\mathcal K,\Xi,\Omega,s_3)$ based solely on the values of $\mathcal K$, $\Xi$, and $s_3$. 
Explicitly, we have:

\begin{lemma}
For fixed $\Xi$, the linearly unstable plane-wave configurations are the values of $(\mathcal K,s_3)$ for which 
\bse
\label{e:Wdiscriminant}
\be
\widetilde\D(\mathcal K,\Xi,s_3)<0, 
\ee
where
\begin{gather}
\widetilde\D(\mathcal K,\Xi,s_3)=\widetilde\D_0^3-\widetilde\D_1^2, 
\\
\widetilde\D_0=4\big(\mathcal K^4-\mathcal K^2(\Xi^2+1)+\Xi^4\big)+8\Xi^2+1, 
\\
\widetilde\D_1=27\mathcal K^2s_3^2+(\mathcal K^2+\Xi^2+1)
\big(8\mathcal K^4-20\mathcal K^2(\Xi^2+1)+8\Xi^2(\Xi^2+2)-1\big)\,. 
\end{gather}
\ese
\end{lemma}

\begin{proof}
We again make use of Lemma~\ref{l:Stevic} and apply it to the polynomial $\widetilde P_{\rm lin}(\mathcal K,\Xi,\Omega,s_3)$ in the variable $\Omega$. The corresponding quantities are 
\be
\P=-16(\mathcal K^2+\Xi^2+1), 
\qquad
\R=-64\big(4\mathcal K^2(\Xi^2+1)+1\big), 
\ee
and the polynomial discriminant is
\be
\operatorname{disc}_{\Omega}\widetilde P_{\rm lin}
    = \tfrac{256}{27}\,\widetilde\D(\mathcal K,\Xi,s_3)\,. 
\ee
Now note the following:
(i) All quantities are even in $\mathcal K$, $\Xi$ and $s_3$. 
This means that we can limit ourselves to study the first quadrant of the $\mathcal K\Xi$-plane or of the $\mathcal K s_3$-plane without loss of generality.
(ii) $\P$ and $\R$ are independent of $s_3$, but $\D$ is not.
(iii) $\P$ and $\R$ are always negative.
Therefore, by Lemma~\ref{l:Stevic}, the existence of complex roots for $\Omega$ 
(which always appear in complex conjugate pairs, and one of which would therefore correspond to unstable modes) is in one-to-one correspondence with the discriminant~$\D$ in~\eqref{e:Wdiscriminant} assuming negative values.
\end{proof}

\begin{figure}[t!]
\centering
\includegraphics[width=0.90\textwidth]{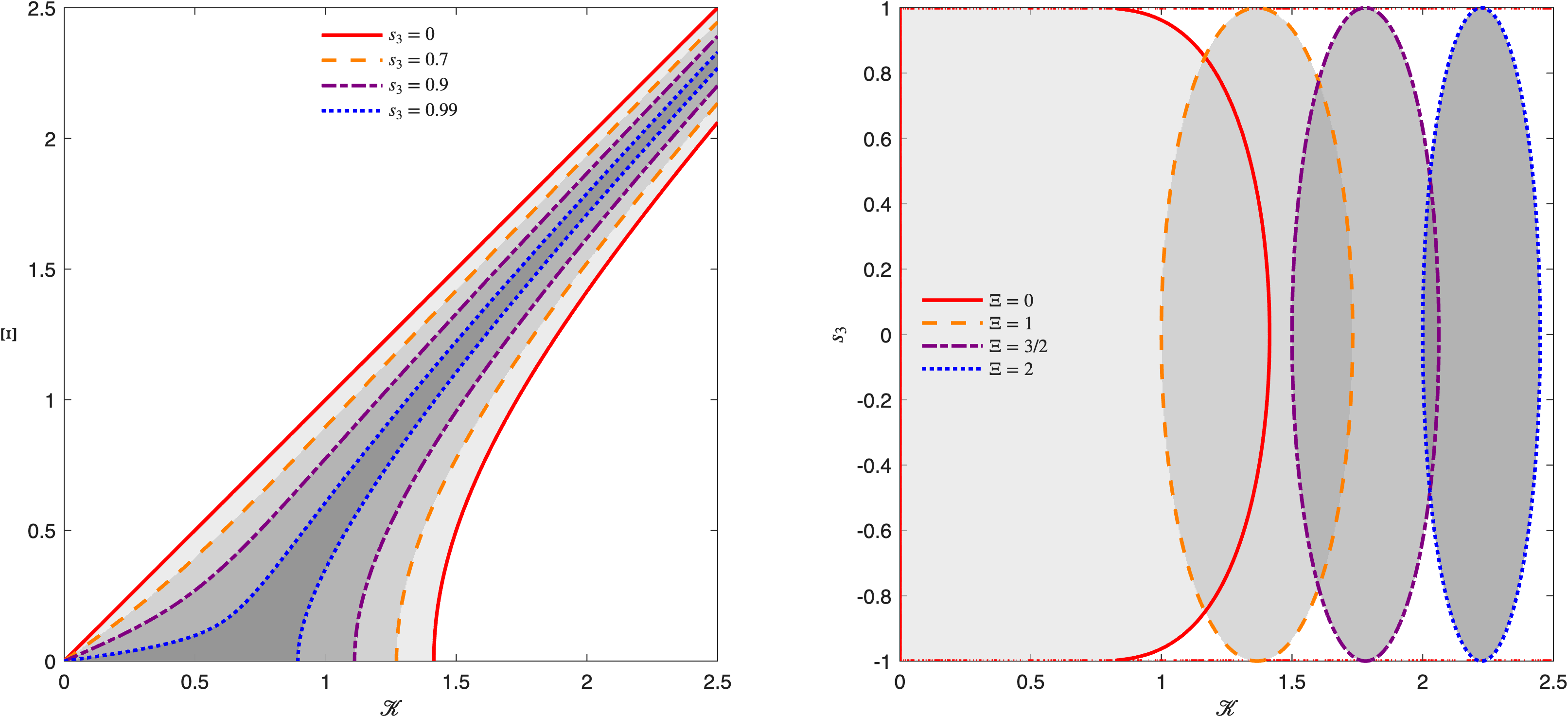}
\caption{Left panel: Instability regions and corresponding thresholds in the $\mathcal K\Xi$-plane, with $\mathcal K$ on the horizontal axis and $\Xi$ on the vertical axis, for $s_3=0$ (red), $s_3=0.7$ (orange), $s_3=0.9$ (purple), and $s_3=0.99$ (blue). In each case, unstable configurations lie in the gray region, shown with progressively darker shading for increasing values of $s_3$. Right panel: Instability regions and corresponding thresholds in the $\mathcal K s_3$-plane, with $\mathcal K$ on the horizontal axis and $s_3$ on the vertical axis, for $\Xi=0$ (red), 1 (orange), 3/2 (purple), and 2 (blue). In each case, unstable configurations lie in the gray region, shown with progressively darker shading for increasing values of $\Xi$.}
\label{f:thresholds}
\end{figure}

The left panel of Fig.~\ref{f:thresholds} shows the instability regions and the corresponding stability thresholds in the $\mathcal K\Xi$-plane for different values of $s_3$, 
while the right panel of Fig.~\ref{f:thresholds} shows the instability regions and the corresponding stability thresholds in the $\mathcal K s_3$-plane for different values of $\Xi$. 

We end this section by showing that, in the long wave limit, the predictions of the above linearized stability analysis coincide exactly with those from Whitham modulation theory.
To this end,
note that, when $\Xi=0$, the linearized quartic and the characteristic-speed quartic satisfy the polynomial identity
\be
\widetilde P_{\rm lin}(W/\sqrt2,0,\sqrt2C,s_3)
=4g(C,W,s_3).
\label{e:longwavepolynomialidentity}
\ee
The discriminant used above factors as
\be
\widetilde\D(\mathcal K,0,s_3)
=27(1-s_3^2)\mathcal K^2\check\D(\mathcal K,s_3), 
\label{e:linearizedDzero}
\ee
where
\be
\check\D(\mathcal K,s_3)=27\mathcal K^2s_3^2+(\mathcal K^2-2)(4\mathcal K^2+1)^2\,. 
\ee
Then, since
$\check\D(W/\sqrt2,s_3)=\tfrac12d(W,s_3)$,
the finite-wavenumber discriminant criterion reduces at $\Xi=0$ to the baseband criterion in Lemma~\ref{l:unstable}.

The identity above determines the exact $\Xi=0$ limit, but it does not resolve the behavior at small nonzero $\Xi$ when the Whitham discriminant vanishes. The latter is described by the following result.

\begin{proposition}
\label{p:degenerate-longwave-instability}
Suppose that $|s_3|<1$ and that
\be
\check\D(\mathcal K,s_3)=0.
\ee
Then, as $\Xi\to0$,
\be
\widetilde\D(\mathcal K,\Xi,s_3)
=
-6(2\mathcal K^2-1)^4(4\mathcal K^2+1)\Xi^2
+O(\Xi^4).
\label{e:degenerate-D-expansion}
\ee
Consequently, for every such point on the nontrivial Whitham boundary, there exists $\Xi_0>0$ such that
\be
\widetilde\D(\mathcal K,\Xi,s_3)<0, 
\qquad
0<|\Xi|<\Xi_0.
\ee
Thus, although the characteristic speeds are all real on this boundary, the corresponding plane waves are linearly unstable for arbitrarily small nonzero perturbation wavenumbers.
\end{proposition}

\begin{proof}
On the boundary $\check\D(\mathcal K,s_3)=0$, one has 
\be
27\mathcal K^2s_3^2=-(\mathcal K^2-2)(4\mathcal K^2+1)^2. 
\ee
Substituting this relation into the expressions for $\widetilde\D_0$ and $\widetilde\D_1$ in~\eqref{e:Wdiscriminant}, and expanding $\widetilde\D=\widetilde\D_0^3-\widetilde\D_1^2$ about $\Xi=0$, gives
\be
\widetilde\D(\mathcal K,\Xi,s_3) = -6(2\mathcal K^2-1)^4(4\mathcal K^2+1)\Xi^2 +O(\Xi^4). 
\ee
The coefficient of $\Xi^2$ is strictly negative on the two-component boundary $|s_3|<1$. Indeed, if $\mathcal K^2=\half$, then 
\be
\check\D(\mathcal K,s_3)=\tfrac{27}{2}(s_3^2-1), 
\ee
so that $\check\D=0$ would require $|s_3|=1$. Hence $\mathcal K^2\ne\half$ under the hypotheses of the proposition, and therefore 
$-6(2\mathcal K^2-1)^4(4\mathcal K^2+1)<0$. 
It follows that $\widetilde\D(\mathcal K,\Xi,s_3)<0$ for all sufficiently small nonzero $\Xi$. 
By the discriminant criterion~\eqref{e:Wdiscriminant}, the linearized quartic then has a complex-conjugate pair of roots, which proves the claim.
\end{proof}

The above discussion
resolves the degenerate case left open by the dispersionless analysis in section~\ref{s:modulationalstability}. 
On the nontrivial boundary $\D=0$, the repeated characteristic speed is exactly double for $|s_3|<1$, and~\eqref{e:degenerate-D-expansion} implies that the corresponding complex-conjugate pair satisfies $\Im\Omega=O(|\Xi|)$ as $\Xi\to0$. Since the temporal growth rate of~\eqref{e:perturbationFourier} is $|\xi\,\Im\omega|$, with $\omega=\sqrt{2\rho}\,\Omega$ and $\xi=\sqrt{2\rho}\,\Xi$, the growth rate on the Whitham boundary is therefore $O(\xi^2)$. 
In the interior of the baseband-unstable region, in contrast, $\Im\Omega$ has a nonzero limit as $\Xi\to0$, giving the usual $O(|\xi|)$ growth. 
Thus the Whitham boundary is unstable for arbitrarily small nonzero perturbation wavenumbers, but with a weaker growth rate than in the interior.
\section{Numerical validation}
\label{s:numerics}

To illustrate the predictions of the previous sections, we integrated numerically the Manakov system~\eqref{e:ManakovISTnormalization} with $\eps=1$ and a variety of initial conditions.
We did so using an eighth-order Fourier split-step method, 2048 spatial grid points, and an integration step size not larger than $3\cdot10^{-4}$.
The initial conditions considered are an exact two-component plane-wave solution $\@q_o(x,0)$ plus a small perturbation, namely:
\bse
\label{e:ICs}
\be
\@q(x,0) = \@q_o(x) + \@q_p(x)\,,
\ee
with 
\be
\@q_o(x) = (A_1\,\e^{\ii Vx},A_2\,\e^{-\ii V x})^\t,
\ee
similarly to~\eqref{e:qOIC}, and with 
\be
\@q_p(x) = \eta\,(\xi_1,\xi_2)^\t,
\label{e:noise}
\ee
\ese
where, at each spatial grid point, $\xi_1$ and $\xi_2$ are taken to be complex numbers with real and imaginary parts given by independent identically distributed normal random variables with zero mean and unit variance. 
In each of the runs described below, the spatial domain was chosen so that the total domain size was a multiple of $2\pi/V$, to ensure the periodicity of the initial condition, and so that
$\Delta x$ was no larger than~0.09.
The dimensionless parameter $\eta $ quantifies the strength of the perturbation.
We used the value $\eta = 0.05$ in all the runs (except for Fig.~\ref{f:seeded} below, which used a different kind of perturbation).

\begin{figure}[t!]
\centering
\includegraphics[width=0.95\textwidth]{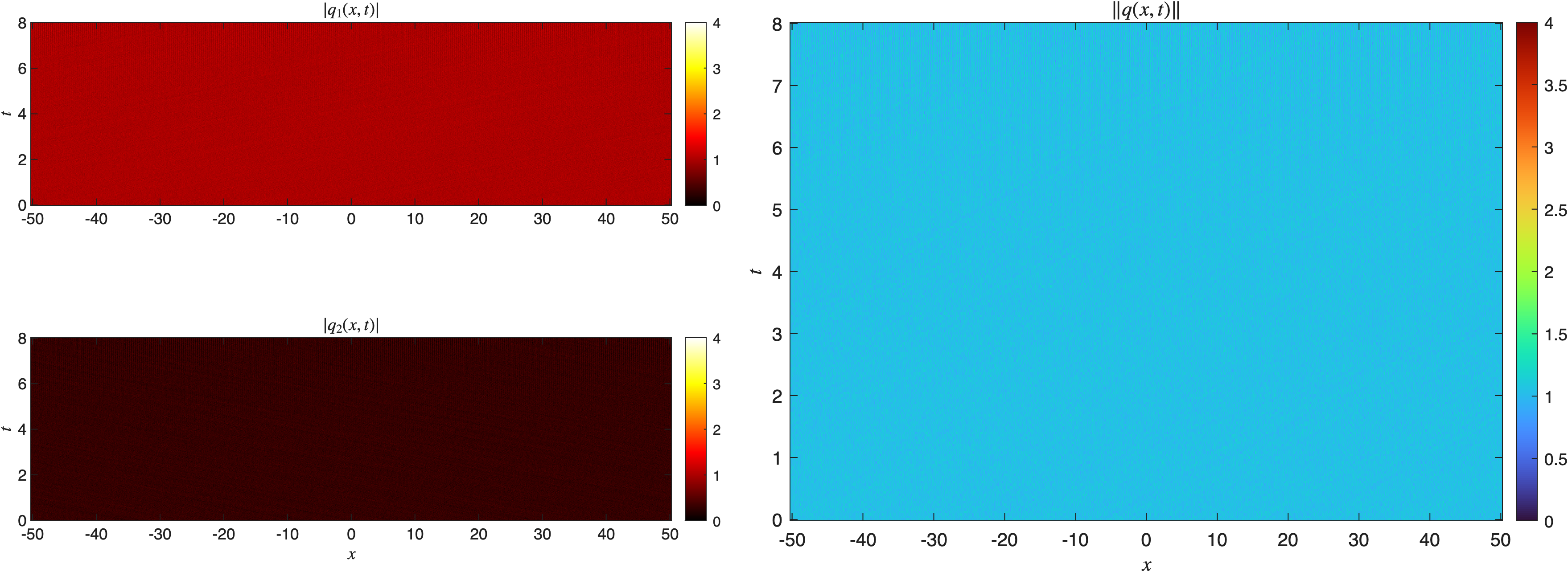}
\caption{Time evolution of an initial condition for the Manakov system consisting of a linearly stable plane wave plus a small perturbation. Left panels: Density plots of $|q_1(x,t)|$ and $|q_2(x,t)|$. Right panel: Density plot of $\|\@q(x,t)\|$. See the text for further details.}
\label{f:stable}
\end{figure}

Figure~\ref{f:stable} shows the evolution of  baseband-stable  initial condition obtained by setting $A_1 = 1$, $A_2=0.3$ and $V = 10$ in~\eqref{e:ICs}.  
In this case, the perturbation shows no appreciable growth over the time interval displayed.

(We should point out that, while Whitham stability analysis of~\ref{s:modulationalstability}
predicts stability with respect to long-wavelength perturbations,
the linearized stability analysis of section~\ref{s:linearization} still predicts instability with respect to perturbations with a narrow range of large wavenumbers.
On the other hand, the corresponding growth rates are sufficiently small that such perturbations do not become $O(1)$ over the time scales considered.)

A very different outcome is produced when the same kind of perturbation is added to  baseband-unstable  plane wave solution.
Figure~\ref{f:unstable} shows the time evolution of an unstable plane wave, obtained by setting $A_1=A_2=1$ and $V=1$ in~\eqref{e:ICs}.
It is evident how, in this case, the small random initial perturbation grows to eventually become comparable with the unperturbed solution.

\begin{figure}[t!]
\centering
\includegraphics[width=0.95\textwidth]{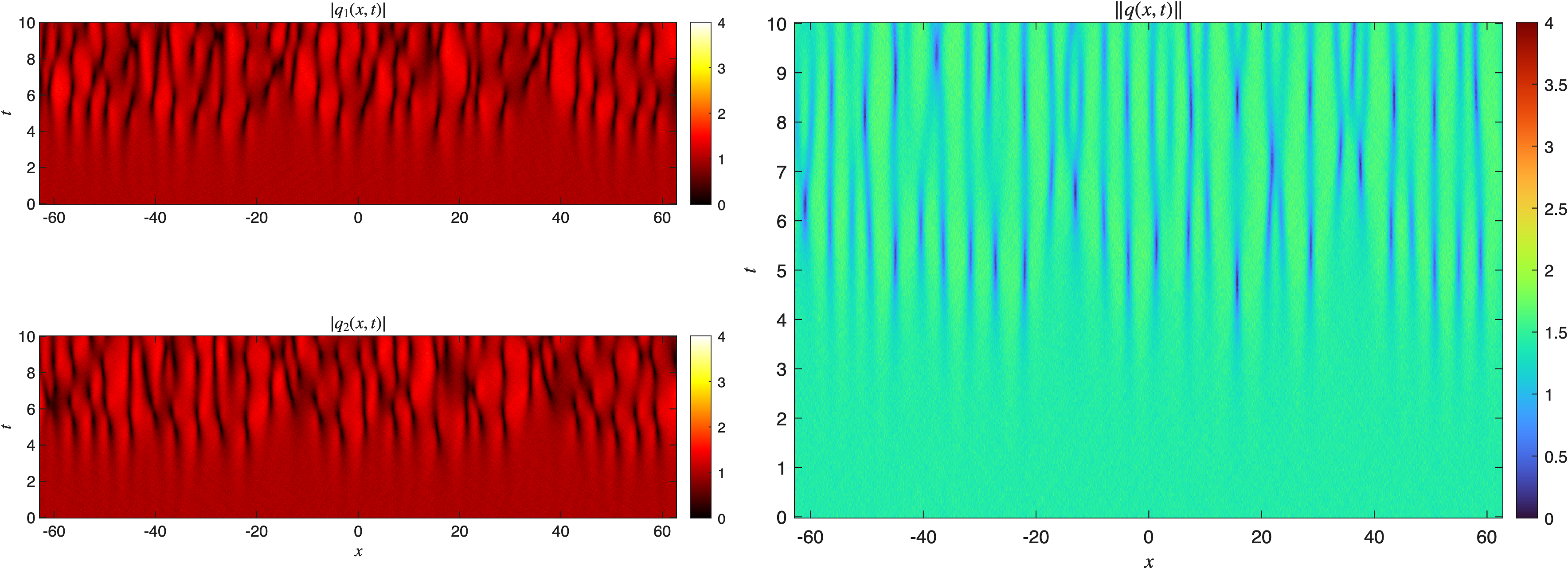}
\caption{Same as Fig.~\ref{f:stable}, but for an unstable initial plane wave configuration. See the text for further details.}
\label{f:unstable}
\end{figure}
\begin{figure}[t!]
\centering
\includegraphics[width=0.95\textwidth]{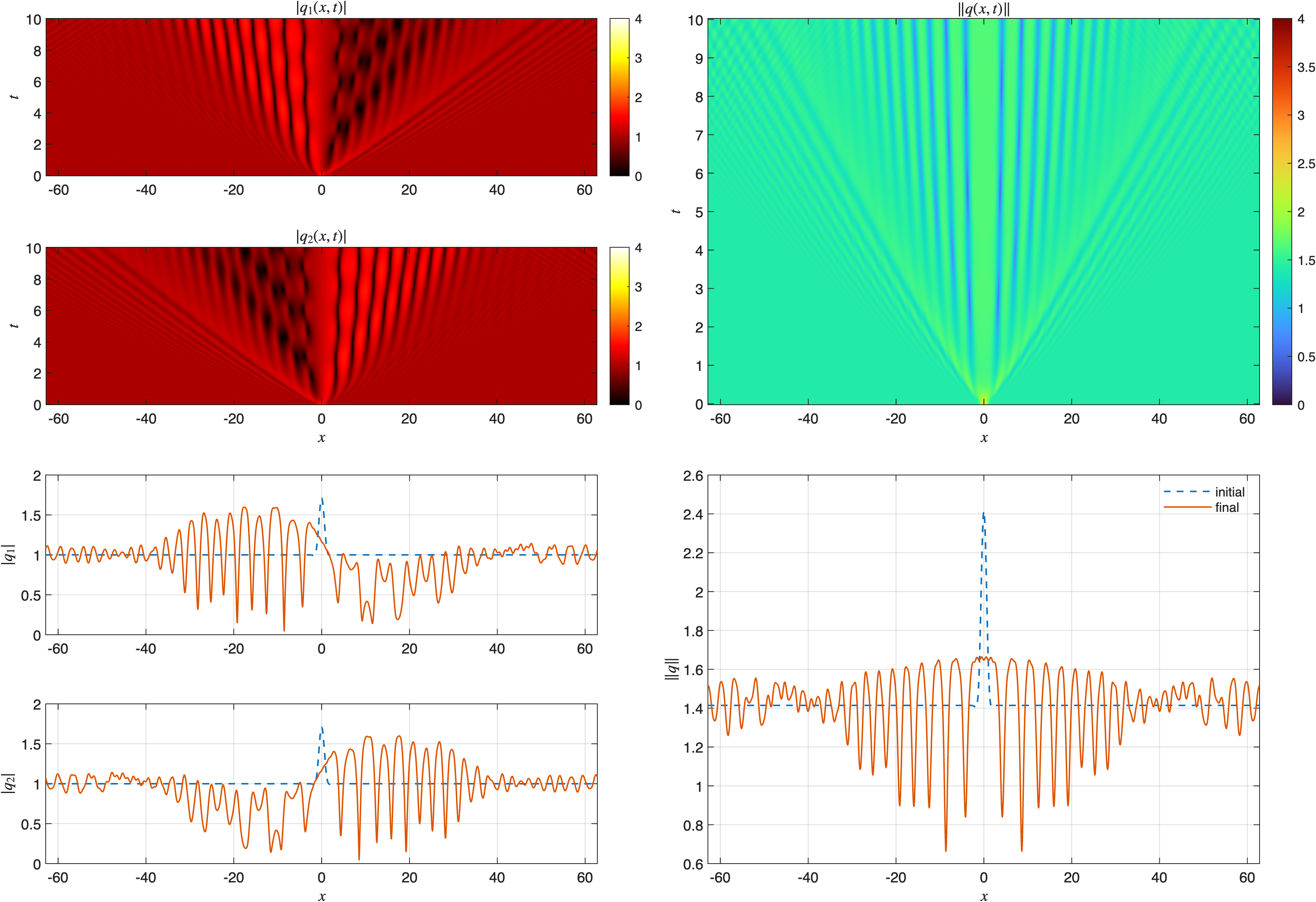}
\caption{Top row: Same as Fig.~\ref{f:unstable}, but with a seeded perturbation. Bottom row: The initial condition (blue) and final profile (red) of the solution. See the text for further details.}
\label{f:seeded}
\end{figure}

A further different outcome is obtained when the random perturbation in Fig.~\ref{f:unstable} is replaced by a localized disturbance. Figure~\ref{f:seeded} shows the evolution from the same unstable plane-wave background, with $A_1=A_2=1$ and $V=1$, but  where $\@q_p(x)$ in~\eqref{e:ICs} is given by 
\be
\@q_p(x)=\e^{-x^2}\big(\e^{\ii\phi}\cos\theta,\e^{-\ii\phi}\sin\theta\big)^\t,
\qquad
\theta=\pi/4,
\qquad
\phi=0,
\ee
 instead of~\eqref{e:noise}. 
No additional amplitude perturbation factor is used in this simulation; 
the perturbation has peak component amplitude $1/\sqrt2$ and peak vector norm equal to one. 
The localized perturbation evolves into a two-component generalization of the dispersive shock wave structures observed in~\cite{BM2016} for the focusing NLS equation, here arising in the defocusing regime. The spatial profiles of the initial and final states are shown in the bottom row of Fig.~\ref{f:seeded}.
Similar two-component structures were also recently observed experimentally in repulsive two-component Bose-Einstein condensates \cite{PRL2025v135p113401}.
In that case, however, the instability was driven by the immiscibility between the two components, as opposed to the relative counterflow as in this case.

\section{Conclusions}
\label{s:conclusions}

In this work we have studied the dispersionless (genus-zero) limit of the defocusing Manakov system.
Starting from a two-fluid Madelung representation, we derived the four-component hydrodynamic-type
Manakov--Whitham system~\eqref{e:ManakovDispless} and showed that it passes the Haantjes tensor test,
consistent with its expected integrability.
We then used the Lax pair of the Manakov system to construct the spectral curve associated with the two-component plane wave solutions, and showed that its branch points are the roots of a quartic (the discriminant of the cubic spectral curve) whose coefficients we computed explicitly in the symmetrized variables $(\rho,s_3,v,w)$.
The central structural result is that each simple branch point of the spectral curve is a genuine local Riemann invariant of the dispersionless Manakov system, with a characteristic speed given by the simple relation $c_j = \lambda_j - k_j$.
Next, using Galilean invariance to formulate the problem to the symmetric frame, we reduced the inversion task on the generic branch $ws_3\ne0$ to two candidate values of $w^2$ together, with the physical admissibility conditions in~\eqref{e:selectionrule}. 
The general $v\ne0$ inversion is treated in~\cite{AB2026}.
The characteristic-speed discriminant gives a complete classification of baseband modulational stability and instability. 
We determined the unstable region in the $(W,s_3)$ plane and proved that the dimensionless baseband growth coefficient decreases with $|s_3|$, with maximum $1/(2\sqrt2)$ at $s_3=0$ and $|W|=\sqrt{3/2}$.
We also characterized the non-strictly-hyperbolic boundary and used direct linearization to describe finite-wavenumber instability, including unstable bands that need not intersect the baseband region.

The significance of these results is twofold. 
On the structural side, Wright~\cite{Wright2013} had already proposed that the branch points of the trigonal spectral curve serve as Riemann invariants for the Manakov Whitham system and confirmed this correspondence in a nontrivial dispersionless example.
Related Riemann-diagonal reductions were subsequently obtained in~\cite{Kamchatnov2014}. Theorem~\ref{t:branchpoints} establishes the corresponding general local result at genus zero for \textit{arbitrary simple ramification points}. 
In particular, the exact identity~\eqref{e:keyidentity} implies directly that $\nabla_{\@y}k_j$ is a left eigenvector of the hydrodynamic coefficient matrix with eigenvalue $c_j=\lambda_j-k_j$, and hence that each simple branch point satisfies the associated Riemann-invariant equation. 
This result provides  a direct, self-contained proof of the branch-point/Riemann-invariant correspondence without recourse to the finite-gap construction.
On the applied side, the explicit and complete stability classification, together with the closed-form expression for the maximal baseband growth coefficient, characterizes the onset of baseband modulational instability of two-component plane waves purely in terms of the physical parameters (total density, polarization imbalance, and counterflow), which is of direct relevance to experiments in nonlinear optics and two-component Bose--Einstein condensates.

Recall that Forest et al.~\cite{JNLS2000v10p291} derived the linearized dispersion relation and analyzed long-wave and intermediate-wavelength instability. Section~\ref{s:linearization} rewrites the same relation as a discriminant criterion and identifies its $\Xi=0$ reduction with the Whitham hyperbolicity boundary. 
Experimental evidence for baseband and passband polarization modulation instability in a defocusing Manakov fiber system was reported in~\cite{Frisquet2015}.
The modulational instability of a related, but non-integrable, two-component NLS system was recently investigated theoretically and experimentally in~\cite{PRL2025v135p113401}.
There, however, the two components were not allowed to counterpropagate, whereas in the present work the counterflow $w$ is precisely the parameter that drives the instability. Indeed, it would be interesting to understand to what extent the mechanism identified here persists in the non-integrable setting.
Finally, we note that the discriminant quartic, its branch points, and the associated modulational instability were examined from a different viewpoint in~\cite{Baronio2014}.
As was discussed in sections~\ref{s:Laxpair} and~\ref{s:modulationalstability}, after the corresponding change of variables the symmetric branch-point quartic and the baseband instability condition agree with those obtained in~\cite{Baronio2014}. The present work extends that analysis by deriving the genus-zero Whitham system and proving the general local branch-point/Riemann-invariant correspondence, addressing the inversion from spectral to physical variables, relating collisions of branch points to collisions of characteristic speeds, determining the monotonicity and maximum of the baseband growth coefficient, and obtaining a discriminant characterization of the full finite-wavenumber instability.

We should mention that an alternative derivation of the Riemann invariants was recently obtained in a separate
paper~\cite{AB2026} without making use of the spectral quartic, and where a covariant formulation of the hydrodynamic
system~\eqref{e:ExactCompContinuity}--\eqref{e:ExactCompVelocity} was also derived.
Similarly, even though the Manakov system is invariant under Galilean transformations, the inversion problem in the
general case $v\ne 0$ is considerably more involved than in the symmetric case: it can ultimately be reduced to finding
a single scalar quantity via the roots of a twelve-degree polynomial.
A parallel analysis, including a setting that avoids this twelve-degree inversion via the signed resolvent cover of the
Kowalevski--Bloch manifold, is developed in~\cite{AB2026}.

The results of this work open the way to a number of follow-up problems.
First, having the dispersionless Manakov system in Riemann-invariant form makes it natural to study Riemann problems for the two-component system, in analogy with the classical theory for scalar hydrodynamics and for the NLS equation. As a byproduct, this would provide a route to the vector generalization of the dam-break problem for the
defocusing NLS equation studied in~\cite{OL20p2291}, as well as to the theory of vector dispersive shock waves.
Such Riemann-problem studies would also connect naturally to the analysis of nonlinear polarization waves and
dispersive shock waves in two-component condensates~\cite{CongyKamchatnovPavloff2016,IvanovKamchatnovCongyPavloff2017}.
Second, a natural next step will be the extension of the present genus-zero analysis to the genus-one (and, more generally, higher-genus) Whitham system, where the modulation of the periodic wave solutions is governed by the full finite-gap spectral data.
Finally, it will be of interest to extend the analysis to the focusing Manakov system and to the $N$-component vector NLS equation, where the interplay between the higher-dimensional spectral curve and the modulation dynamics is largely unexplored.
We plan to address some of these questions in the future.

\bigskip
\paragraph{Acknowledgments.}
This work originated from conversations between G.B.\ and Mark Hoefer in which the idea of associating the Riemann invariants of the dispersionless Manakov system to the branch points of the spectral curve of plane wave solutions of the Manakov system was initially conceived, and we
thank Mark Hoefer for many insightful discussions on this topic.
We also thank Gennady El, Amrita Ganapathy and Barbara Prinari for many interesting conversations on various topics related to the present work.  J.A. was partially supported by NSF grant PHY-2316622.
G.B.\ was partially supported by the Simons Foundation under grant number SFI-MPS-TSM-00013369.

\addcontentsline{toc}{section}{References}
\def\reftitle#1{``#1''}
\def\booktitle#1{\textit{#1}}

\end{document}